\documentclass[journal,12pt,onecolumn,draftclsnofoot,]{IEEEtran}%
\usepackage[left=2.1cm, right=2.1cm, top=2.0cm, bottom = 2.0cm]{geometry}
\usepackage{amsmath}
\usepackage{amssymb}
\usepackage{amsthm}
\usepackage{graphicx}
\usepackage{mathtools}
\usepackage[ruled, vlined, linesnumbered]{algorithm2e}
\usepackage{color,soul}

\usepackage{cite}
\usepackage{subfigure}
\usepackage{stmaryrd}
\usepackage{breqn}
\usepackage{stackengine}
\usepackage{verbatim} 
\usepackage{glossaries}
\usepackage{color,soul}
\usepackage{xcolor}

\newtheorem{theorem}{\textit{Theorem}}

\DeclareMathOperator*{\argmin}{\rm argmin}
\usepackage{xspace}
\newcommand{\MATLAB}{\textsc{Matlab}\xspace}
\newcommand{\nonl}{\renewcommand{\nl}{\let\nl\oldnl}}
\makeatletter
\newcommand\approxsim{\mathchoice
  {\@approxsim {\displaystyle}      {1ex} }
  {\@approxsim {\textstyle}         {1ex} }
  {\@approxsim {\scriptstyle}       {.7ex}}
  {\@approxsim {\scriptscriptstyle} {.5ex}}}
\newcommand\@approxsim[2]{%
  \mathrel{%
    \ooalign{%
      $\m@th#1\sim$\cr
      \hidewidth$\m@th#1.$\hidewidth\cr
      \hidewidth\raise #2 \hbox{$\m@th#1.$}\hidewidth\cr
    }%
  }%
}
\makeatother

\ifCLASSINFOpdf
 
\else
  
\fi

\begin{document}

\title{Reconfigurable Intelligent Surface Enabled Full-Duplex/Half-Duplex Cooperative Non-Orthogonal Multiple Access}
\author{Mohamed Elhattab,~\IEEEmembership{Student Member,~IEEE}, \IEEEauthorblockN{Mohamed Amine Arfaoui},~\IEEEmembership{Student Member,~IEEE}, Chadi Assi,~\IEEEmembership{Fellow,~IEEE}, Ali Ghrayeb,~\IEEEmembership{Fellow,~IEEE} \thanks{Mohamed Elhattab is with the ECE Department, Concordia University, Montreal, Quebec, H3G 1M8, Canada (email: m\_elhatt@encs.concordia.ca). Mohamed Amine Arfaoui and Chadi Assi are with the CIISE Department, Concordia University, Montreal, Quebec, H3G 1M8, Canada (email: m\_arfaou@encs.concordia.ca, assi@mail.concordia.ca). Ali Ghrayeb is with the Electrical and Computer Engineering (ECE) department, Texas A \& M University at Qatar, Doha, Qatar (e-mail: ali.ghrayeb@qatar.tamu.edu).}}
\maketitle
\begin{abstract}
This paper investigates the downlink transmission of reconfigurable intelligent surface (RIS)-aided cooperative non-orthogonal-multiple-access (C-NOMA), where both half-duplex (HD) and full-duplex (FD) relaying modes are considered. The system model consists of one base station (BS), two users and one RIS. The goal is to minimize the total transmit power at both the BS and at the user-cooperating relay for each relaying mode by jointly optimizing the power allocation coefficients at the BS, the transmit power coefficient at the relay user, and the passive beamforming at the RIS, subject to power budget constraints, the successive interference cancellation constraint and the minimum required quality-of-service at both cellular users. To address the high-coupled optimization variables, an efficient algorithm is proposed by invoking an alternating optimization approach that decomposes the original problem into a power allocation sub-problem and a passive beamforming sub-problem, which are solved alternately. For the power allocation sub-problem, the optimal closed-form expressions for the power allocation coefficients are derived. Meanwhile, the semi-definite relaxation approach is exploited to tackle the passive beamforming sub-problem. The simulation results validate the accuracy of the derived power control closed-form expressions and demonstrate the gain in the total transmit power brought by integrating the RIS in C-NOMA networks.
\end{abstract}
\vspace{-0.4cm}
\begin{IEEEkeywords}
C-NOMA, FD, HD,  passive beamforming, power allocation, RIS, 6G.   
\end{IEEEkeywords}
\section{Introduction}
\subsection{Motivation}
With the proliferation of numerous burgeoning applications and services such as ultra-reliable low-latency communication (URLLC), massive machine type communications (mMTC), enhanced mobile broadband (eMBB), among others, wireless communication systems are expected to face daunting challenges\cite{Vaezi_2019}. To satisfy and accommodate these ever-increasing traffic demands, the diverse quality-of-services (QoS) requirements, and the massive connectivity accompanied with these new applications, various innovative and promising technologies and architectures will need to be developed \cite{Akyildiz_2020}. Novel multiple-access-techniques are currently been explored in both academia and industry lately \cite{Ding_2017_A_Survey} in order to accommodate for such unprecedented requirements. \par Non-orthogonal multiple access (NOMA) has been deemed as one of the vital enabling multiple access techniques for the upcoming sixth-generation (6G) cellular networks \cite{Vaezi_2019, Elhattab_TCOM}. This is due to its ability of enhancing the network spectral efficiency and supporting massive number of connected devices \cite{Vaezi_2019}. NOMA has been included in the third Generation Partnership Project (3GPP) LTE-A standard (Release 15) \cite{Ding_2020}.\footnote{The focus of this  paper is on power domain NOMA, and thus, throughout this  paper, the term NOMA is used to refer to power domain NOMA, unless otherwise stated.} Different from the traditional orthogonal multiple access (OMA) techniques, NOMA allows multiple user equipment (UEs) to utilize the same resource (radio frequency, time slot, code, etc.). This can be achieved by using superposition coding at the transmitter with appropriately power allocation coefficients and performing successive interference cancellation (SIC) at the receiver side to remove the intra-NOMA user interference \cite{Elhattab_TCOM}. This allows the UEs with poor channel conditions, which are referred to as ``far'' NOMA UEs, to be simultaneously served  with the UEs that have good channel conditions, which are referred to as ``near'' NOMA UEs. It has been proven that NOMA achieves a significant gain over OMA techniques \cite{Zhang_2017_Downlink, Liu_2017_Non}.
\par With the objective of further enhancing the performance of NOMA cellular systems, the integration of user-cooperative relaying and NOMA, which is known as cooperative NOMA (C-NOMA), has been proposed \cite{Zeng_2020_Cooperative}. In this transmission technique, with the aid of device-to-device (D2D) communication and the half-duplex (HD)/full-duplex (FD) relaying modes, the near NOMA UEs act as decode-and-forward (DF) relays for the far NOMA UEs. Specifically, the near NOMA UEs exploit the successive decoding strategy adopted at their sides to assist the transmissions between the BS and the far NOMA UEs \cite{Elhattab2020AJoint, Elhattab_2020_Power}. As a result, by adding the cooperation link between the near NOMA UEs and the far NOMA UEs as a new degree of diversity with the BS-UE links, C-NOMA based cellular network can achieve higher spatial diversity and better fairness than NOMA \cite{Elhattab2020AJoint}. While NOMA/C-NOMA transmission techniques are expected to unleash the potential of next-generation cellular networks, another new degree of freedom can be added to the network by making the wireless propagation environment controllable and programmable, which can be achieved through the reconfigurable intelligent surface (RIS) technology \cite{Sena_2020_What}.
\par RIS has recently been recognized as a key promising technology for achieving cost-, energy-, and spectral-efficient communications via intelligently reconfiguring and controlling the wireless propagation environment \cite{Gong_2020}, \cite{Basar_2019}. It is composed of a large number of passive low-cost elements, each of which is capable of independently tuning the phase-shift of the incident radio waves. For instance, by appropriately configuring the phase-shifts with the aid of the RIS controller, the reflected signals can be made to add constructively at the receiver, and therefore, enhancing the received signal strength at the point of interest. In contrast to traditional relays, RIS has the benefit of low power consumption as it passively reflects the incident signals without the need of any radio frequency chains. In addition, the radio signal reflected by the RIS is free from noise corruption or self-interference (SI) in an inherently full-duplex fashion \cite{Basar_2019}. Motivated by the aforementioned beneﬁts of both the RIS and C-NOMA, the potential performance enhancement brought by effectively integrating RIS technology with C-NOMA-based cellular networks is investigated in this paper. This amalgamation between RIS and C-NOMA technologies can provide a promising paradigm for the upcoming 6G networks by providing additional paths that can jointly construct a stronger combined channel gain at the user of interest by leveraging the RIS technology and by improving the network connectivity by adopting the C-NOMA technique.
\subsection{State-of-Art}
The research on RIS-enabled NOMA-based cellular networks is gaining momentum to enhance and improve different performance metrics, such as power consumption \cite{Zheng_2020_Intelligent, Wang_2020_On, fu2019reconfigurable}, network spectral efficiency \cite{Guo_2020_Intelligent, MU_2020Exploiting, ni_2020_resource}, network energy efficiency \cite{Fang_2020_Energy},  user fairness \cite{yang2020intelligent}, and physical layer security\cite{Zhang_2020_Robust}. The authors in \cite{Zheng_2020_Intelligent} evaluated the minimum power consumption in a two-UE NOMA network for three different multiple  access  schemes (frequency division multiple access, time division multiple access, and NOMA) to achieve the same required data rate threshold, and have proved that NOMA can achieve a superior performance compared to the counterpart OMA schemes in most of the system settings. The transmit beamforming at the BS, phase-shift matrix at the RIS and the channel gains ordering for a multi-UE NOMA cellular network were jointly optimized to minimize the total transmit power in \cite{fu2019reconfigurable}. With the goal of maximizing the sum-rate in a two-UE and multi-UE NOMA cellular network, a joint optimization of the power allocation at the BS and the passive beamforming at the RIS were investigated in \cite{Guo_2020_Intelligent,MU_2020Exploiting}, respectively. Aiming at improving the spectral efficiency of a multi-cell RIS-NOMA network, the  problem  of joint user association, sub-channel  allocation, power control, and passive beamforming was formulated in \cite{ni_2020_resource}. By jointly optimizing the active beamforming at the BS and the passive beamforming at the RIS, an efficient algorithm was developed in \cite{Fang_2020_Energy} to maximize the network energy efficiency of a two-user RIS-NOMA network. A joint active beamforming at the BS and passive beamforming at the RIS was investigated in \cite{yang2020intelligent} to maximize the minimum achievable data rate, and accordingly, ensure user fairness. 
\par Going deep into the investigation of RIS-NOMA cellular networks, the performance analysis  using tools from stochastic geometry was considered in \cite{Elhattab_2020, Hou_2020_Reconfigurable, Ding_2020_Simple, zhang2020reconfigurable}. Specifically, the authors in \cite{Elhattab_2020} analyzed the network spectral efficiency when both the RIS technology and the coordinated multipoint transmission were jointly integrated in NOMA cellular networks. The network spectral efficiency, energy efficiency, and outage probability were derived in \cite{Hou_2020_Reconfigurable}. In \cite{Ding_2020_Simple}, the RIS is used for assisting NOMA transmission when spatial division multiple access was applied by the BS to generate orthogonal beams for serving multiple two-UE NOMA clusters. Closed form expressions for the coverage probability and the ergodic rate were derived in an RIS assisted two-UE NOMA network in \cite{zhang2020reconfigurable}. 
\par All the aforementioned research works \cite{Zheng_2020_Intelligent, Wang_2020_On, fu2019reconfigurable, Guo_2020_Intelligent, MU_2020Exploiting, ni_2020_resource, Fang_2020_Energy, yang2020intelligent, Zhang_2020_Robust, Elhattab_2020, Hou_2020_Reconfigurable, Ding_2020_Simple, zhang2020reconfigurable} have only investigated the performance of integrating RIS with NOMA without considering the user-relaying cooperation, whilst there is a lack of investigations in the existing literature on the performance of C-NOMA when the RIS is considered in the network. Recently, the authors in \cite{zuo2020reconfigurable} studied the performance of an RIS-assisted HD C-NOMA system by jointly optimizing the active beamforming at the BS, the user relaying power, and the phase shifts at the RIS. However, this work has not comprehensively analyzed the performance of RIS-aided HD C-NOMA compared to the FD C-NOMA, either with and without the RIS technology.
\subsection{Contributions}
Against the above background, and to the best of our knowledge, integrating RIS with C-NOMA in HD relaying mode has not been well studied in the literature and remains still unexplored. Moreover, this paper is one of the early attempts to explore the performance of FD C-NOMA systems with the assistance of RIS. These facts motivate this article to study the joint power allocation at the BS, transmit relaying power at the near UE, and passive beamforming at the RIS that minimize the total transmit power. Driven by the aforementioned observations, the main contributions of this paper can be summarized as follows.
\begin{itemize}
\item  A downlink RIS-enabled HD/FD C-NOMA framework consisting of one BS, one near NOMA user, one far NOMA user, and one RIS is considered in this paper, where the near NOMA user can relay the message of the far NOMA user in either a HD or a FD relaying mode. For each relaying mode, this framework is formulated as an optimization problem with the objective of minimizing the total transmit power while guaranteeing the data rate QoS requirements for the users, the power budget at both BS and near NOMA user and the SIC constraint.
\item Due to the high coupling between the power allocation coefficients at both the BS and near NOMA user from one side and the passive beamforming at the RIS from the other side, the formulated total transmit power minimization problem is neither linear nor convex, and hence, is difficult to be directly solved. In order to overcome this challenge, the alternating optimization approach is adopted, in which the original optimization problem is decomposed into two sub-problems, namely, a power allocation sub-problem and a phase-shift optimization sub-problem, which are solved in an alternating manner.
\item For the power allocation sub-problem, and for given phase-shift matrices at the RIS, the feasibility conditions as relations between the QoS requirements, the power budget of the active nodes (BS and near NOMA user), and SIC constraints are derived. Then, the optimal solution of the power allocation sub-problem is determined in closed-form expressions, i.e., with a computational complexity of $\mathcal{O}(1)$. 
\item For the phase-shift optimization sub-problem, and for given values of the power allocation coefficients at the BS and fixed transmit relaying power at the near NOMA user, semi-definite relaxation (SDR) technique is applied to obtain a solution of the phase-shift matrix at the RIS.
\end{itemize}
\par Extensive simulations were carried out to evaluate the performance of the proposed RIS-enabled HD/FD C-NOMA cellular network. In order to validate the effectiveness of the proposed framework, the FD C-NOMA scheme without the assistance of the RIS was considered as a benchmark. The numerical results demonstrate  the efficacy of the proposed framework compared to the FD C-NOMA without RIS scheme in terms of the total transmit power. In fact, the numerical results unveil that 
\begin{itemize}
    \item The RIS-enabled FD C-NOMA provides a significant gain in terms of the total transmit power compared to the RIS-enabled HD C-NOMA and the FD C-NOMA without RIS despite the existence of high residual SI at near NOMA user.
    \item The FD C-NOMA with the assistance of RIS has more resistance to the residual SI effect and can tolerate high SI values compared to the same system without RIS.
    \item RIS-enabled HD C-NOMA can beat the FD C-NOMA without RIS scheme despite the pre-log penalty in the HD relaying mode. This performance gain depends on the number of RIS reflecting elements, the SI channel gain at the near NOMA user, and the required QoS at the NOMA users.
\end{itemize}   
\begin{table}[t]
\centering
\caption{Table of Notations}
\label{Table: Symbols}
\renewcommand{\arraystretch}{0.55} 
\setlength{\tabcolsep}{0.2cm} 
\begin{tabular}{| l | l |}
  \hline 
  \multicolumn{2}{|c|}{\textbf{System Parameters}} \\
  \hline 
  \hline
  $P_{\mathrm{BS}}$, $P_{\rm n}$ & Power budget of BS and UE$_{\rm n}$, respectively. \\ 
  \hline
  $R_{\rm n}^{\mathrm{th}}$ and $R_{\rm f}^{\mathrm{th}}$ & Minimum required data rate for UE$_{\rm n}$ and UE$_{\rm f}$, respectively.\\
  \hline
  $\beta^{\mathrm{HD}}, \beta^{\mathrm{FD}}$ & Power fraction coefficients at UE$_{\rm n}$ in HD and FD relaying mode, respectively. \\ 
  \hline
  $\alpha^{\mathrm{HD}}_{\rm n}, \alpha^{\mathrm{HD}}_{\rm f}$ & Power allocation coefficients at BS for UE$_{\rm n}$ and UE$_{\rm f}$ in the HD case, respectively.  \\ 
  \hline
  $\alpha^{\mathrm{FD}}_{\rm n}, \alpha^{\mathrm{FD}}_{\rm f}$ & Power allocation coefficients at BS for UE$_{\rm n}$ and UE$_{\rm f}$ in the FD case, respectively.  \\  
  \hline 
  $\boldsymbol{\Theta}_{(1)}^{\mathrm{HD}}$, $\boldsymbol{\Theta}_{(2)}^{\mathrm{HD}}$  & RIS phase-shift matrix in the first and second time slots, respectively, in the HD case. \\ 
  \hline 
  $\boldsymbol{\Theta}^{\mathrm{FD}}$  & RIS phase-shift matrix in the FD case. \\
  \hline 
  $\boldsymbol{h}_{\rm br}, \boldsymbol{h}_{\rm rn}, \boldsymbol{h}_{\rm rf}, \boldsymbol{h}_{\rm nr}$ & Channel gains for BS-RIS, RIS-UE$_{\rm n}$, RIS-UE$_{\rm f}$, and UE$_{\rm n}$-RIS links, respectively. \\
  \hline 
  $h_{\rm bn}, h_{\rm bf}, h_{\rm nf}$ & Channel gains for BS-UE$_{\rm n}$, BS-UE$_{\rm f}$, UE$_{\rm n}$-UE$_{\rm f}$ links, respectively.\\
  \hline 
  $\hat{\boldsymbol{h}}_{\rm rf}$ & Channel gain for RIS-UE$_{\rm f}$ link in the second time-slot for the HD case.  \\ 
  \hline 
  $\eta_x$ & Path-loss exponent for the communication link $x$. \\ 
  \hline
\end{tabular} 
\label{T1}
\end{table}
\vspace{-0.8cm}
\subsection{Paper Outline and Notations}
The rest of the paper is organized as follows. Section \ref{Sec:System Model} presents the system model of  the proposed RIS-enabled C-NOMA system. Section \ref{SINR and Rate} presents the signal-to-interference-plus-noise-ratio ($\tt{SINR}$) and the rate analysis for both HD and FD relaying modes. Section \ref{Sec: HD-CNOMA RIS} and \ref{Sec:FD C-NOMA RIS} present the formulated optimization problem and the solution approach for the RIS-enabled HD C-NOMA system and the RIS-enabled FD C-NOMA system, respectively. The simulation results and the conclusion are presented in Sections \ref{Sec:Simulation} and \ref{Sec:Conclusion}, respectively.
\par The notations and symbols adopted throughout the paper are summarized in Table \ref{Table: Symbols}. In addition, vectors and matrices are denoted by bold-face lower-case and upper-case letters, respectively. The distribution of a circularly symmetric complex Gaussian random vector with mean vector $\boldsymbol {x}$ and covariance matrix ${\boldsymbol \Sigma}$ is denoted by $\mathcal {CN}(\boldsymbol {x},{\boldsymbol \Sigma })$. For a complex-valued vector $\boldsymbol {y} , |\boldsymbol {y}|$ accounts for its Euclidean norm, $\arg (\boldsymbol {y})$ denotes a vector that contains the angles of the elements of $\boldsymbol {y}$, and $\text {diag}(\boldsymbol {y})$ denotes a diagonal matrix in which each element in the diagonal is the corresponding element in $\boldsymbol {y}$. For a square matrix $\boldsymbol S , {\mathrm{tr}}(\boldsymbol S)$ denotes its trace, while $\boldsymbol S\succeq \boldsymbol {0}$ means that $\boldsymbol S$ is positive semi-definite. For any general matrix $\boldsymbol {M}, \boldsymbol {M}^{H}$ and  ${\mathrm{rank}}(\boldsymbol {M})$ denote its conjugate transpose and rank, respectively. Moreover, $e^{(\cdot)}$ and $\exp(\cdot)$ denote the exponential function. Finally, $[\boldsymbol{h}]_{n}$ is the $n$th entry of $\boldsymbol{h}$.
\vspace{-0.5cm}
\section{System Model}
\label{Sec:System Model}
\subsection{Network Model}
We consider a downlink transmission in an RIS-enabled two-UE C-NOMA cellular system, which consists of one BS, one near UE denoted by UE$_{\rm n}$, one far UE denoted by UE$_{\rm f}$ and one RIS equipped with $M$ reflecting elements as shown in Fig. \ref{System_Model}. In the proposed model, UE$_{\rm n}$ has, on average, a better channel condition than UE$_{\rm f}$, since UE$_{\rm n}$ is considered as the near UE and UE$_{\rm f}$ is deemed as the far UE. Therefore, in order to enhance the signal quality at the far NOMA UE, RIS is deployed near to UE$_{\rm f}$. With the assistance of the RIS, the BS serves UE$_{\rm n}$ and UE$_{\rm f}$ simultaneously using NOMA, while UE$_{\rm n}$ relays the signal for UE$_{\rm f}$ either in a DF HD relaying mode or in a DF FD relaying mode. Let ${h}_{\rm bn} \in \mathbb{C}, {h}_{\rm bf} \in \mathbb{C}, \boldsymbol{h}_{\rm br} \in \mathbb{C}^{M \times 1}, \boldsymbol{h}_{\rm rn} \in \mathbb{C}^{M \times 1}, \boldsymbol{h}_{\rm nr} \in \mathbb{C}^{M \times 1}, \boldsymbol{h}_{\rm rf} \in \mathbb{C}^{M \times 1}, h_{\rm nf} \in \mathbb{C}$, and $h_{\mathrm{SI}}$ be the channel coefficients of the communication links from BS $\longrightarrow$ UE$_{\rm n}$, BS $\longrightarrow$ UE$_{\rm f}$, BS $\longrightarrow$ RIS, RIS $\longrightarrow$ UE$_{\rm n}$, UE$_{\rm n}$ $\longrightarrow$ RIS, UE$_{\rm n}$ $\longrightarrow$ UE$_{\rm f}$, and of the SI link, respectively. Apart from $h_{\mathrm{SI}}$, the small-scale fading and the large-scale fading are considered both for each communication link. Nevertheless, for all $y \in \{{\rm bn}, {\rm bf}, \rm nf\}$ and $z \in \{{\rm rn}, \rm nr\}$, the small scale fading of $h_{y}$ and $\boldsymbol{h}_{z}$ are modeled as Rayleigh fading. Consequently, the corresponding channel coefficients can be expressed as  
\begin{figure}[!t]
     \centering
     \includegraphics[width = 0.55 \columnwidth]{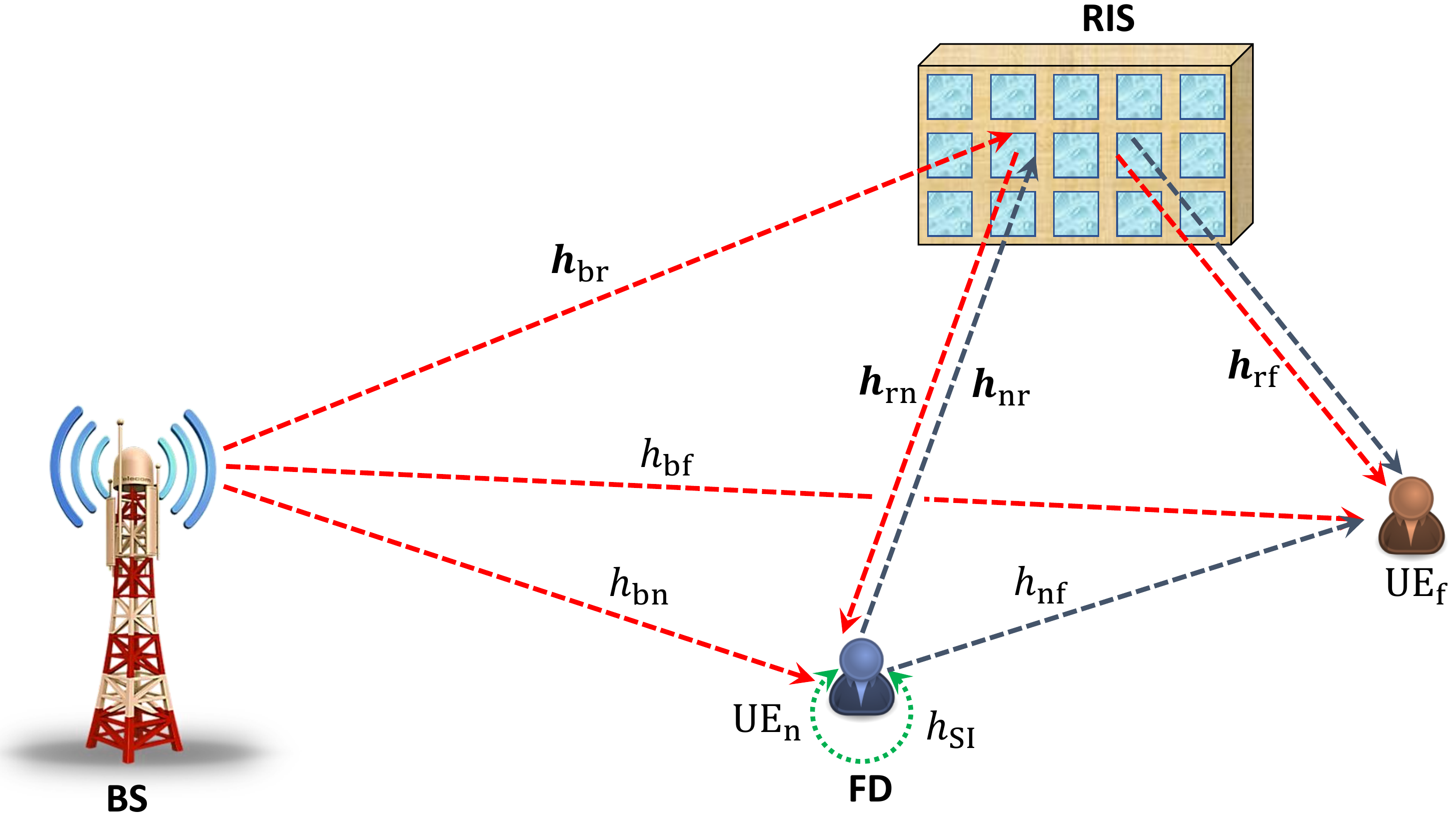}
     \caption{RIS-enabled HD/FD C-NOMA cellular network.}
     \label{System_Model}
 \end{figure}
\begin{align}
    {h}_{y} &= g_{y} \sqrt{PL(d_y)}, \qquad \forall\,y \in \{{\rm bn}, {\rm bf}, \rm nf\}, \cr
    \boldsymbol{h}_{z} &= \textbf{g}_{z} \sqrt{PL(d_z)}, \qquad \forall\,z \in \{{\rm rn}, \rm nr\},
\end{align}
where $g_y \in \mathbb{C}$ and $\textbf{g}_z \in \mathbb{C}^{M \times 1}$ are the small-scale Rayleigh fading with zero mean and unit variance, and $PL(d_y)$ and $PL(d_z)$ are the large-scale path-losses modelled, respectively, as $PL(d_y) = \rho_0 (\frac{d}{d_0})^{-\eta_y}$ and $PL(d_z) = \rho_0 (\frac{d}{d_0})^{-\eta_z}$, such that $\rho_0$ is the path-loss at a reference distance $d_0$, $\eta_y$ and $\eta_z$ are the path-loss exponents, and $d_y$ and $d_z$ are the distance between the end-to-end nodes in the $y$th and $z$th links, respectively.  
\par Considering the communication links between the BS and the RIS and between the RIS and UE$_{\rm f}$, line-of-sight (LoS) components are assumed to exist \cite{Guo_2020_Intelligent}. Thus, these communication links experience small-scale fading that are modelled as Rician fading. The corresponding channel coefficients for these communication links can be expressed as,
\begin{equation}
    \boldsymbol{h}_{x} =  \sqrt{PL(d_x)}\left(\sqrt{\frac{1}{1+\kappa_x}}\textbf{g}_{x} + \sqrt{\frac{\kappa_x}{1+\kappa_x}}\hat{\textbf{g}}_{x} \right), \qquad \forall\,x \in \{\rm br, rf\},
\end{equation}
where $\kappa_x$ is the Rician factor, $\hat{\textbf{g}}_x$ denotes the deterministic line-of-sight component and ${\textbf{g}}_x$ represents the non-line-of-sight component, which follows a Rayleigh distribution with mean zero and variance one. Due to the double fading effects of the BS-RIS-UE cascaded link, the cascaded communication link suffers from more severe path-loss than the direct communication link, i.e., the BS-UE link. Thus, it is assumed that the power of the reflected signals by the RIS two or more times can be ignored.  The channel state information (CSI) of all the communication links is assumed to be perfectly known at the BS \cite{Zheng_2020_Intelligent, Wang_2020_On, fu2019reconfigurable, Guo_2020_Intelligent, MU_2020Exploiting, ni_2020_resource, Fang_2020_Energy, yang2020intelligent,  Elhattab_2020, Hou_2020_Reconfigurable, Ding_2020_Simple, zhang2020reconfigurable}. \footnote{To characterize the theoretical performance gain and provide useful insights achieved by the amalgamation between RIS and C-NOMA, we assume that the CSI of all communication links is available at the BS. This CSI can be efficiently obtained by one of the channel estimation techniques for RIS-assisted wireless networks \cite{You_2020_Channel, taha2019enabling}.} \\ 
\indent Since UE$_{\rm n}$ can adopt two different relaying modes, i.e., either HD or FD, two different transmission schemes, i.e., RIS-enabled HD C-NOMA and RIS-enabled FD C-NOMA, can be adopted, which are adequately defined in the next subsection.
\vspace{-0.4cm}
 \subsection{Transmission Model}
 The transmission model for the two-UE C-NOMA consists of two phases (direct transmission phase and cooperative transmission phase) that are detailed as follows.
 \begin{itemize}
\item \textbf{Direct transmission (DT) phase}: the BS applies superposition coding on the signals intended to UE$_{\rm n}$ and UE$_{\rm f}$. Then, it transmits the superimposed signal to both of them. The transmitted signal by the BS will hit the RIS and will be then reflected back to UE$_{\rm n}$ and UE$_{\rm f}$ to improve their signals reception diversity. Following NOMA principle, UE$_{\rm n}$ first performs SIC to decode the intended signal for UE$_{\rm f}$. Second, it cancels the decoded signal of the UE$_{\rm f}$ from its own reception. Afterwards, UE$_{\rm n}$ decodes its own signal from the resulting reception. Meanwhile, UE$_{\rm f}$ treats the signal of UE$_{\rm n}$ as a noise. 
\item \textbf{Cooperative transmission (CT) phase}: UE$_{\rm n}$ relays the decoded signal of UE$_{\rm f}$ through a D2D channel. Consequently, two signals will be received at the UE$_{\rm f}$'s side. The first is resulting from the transmission of UE$_{\rm n}$ and the second signal comes from the reflection by the RIS. Finally, UE$_{\rm f}$ combines the received signals coming from the BS, RIS, and UE$_{\rm n}$ and then decodes its own signal.
 \end{itemize}
 \par For the case of HD relaying mode, the two transmission phases (DT and CT) occur in two consecutive time-slots. As shown in Fig. \ref{System_Model}, the red lines represent the DT phase, which occurs in the first time-slot and the black lines show the CT phase, which occurs in the second time-slot. However, for the case of FD relaying mode, DT and CT occur in the same time-slot with the cost of inducing SI at UE$_{\rm n}$ (the green line in Fig. \ref{System_Model}). After presenting the main operations of RIS-enabled HD C-NOMA and RIS-enabled FD C-NOMA, we direct our attention toward calculating the received $\tt{SINR}$ and the corresponding achievable data rates at UE$_{\rm n}$ and UE$_{\rm f}$ in the two different operation modes (FD and HD) as shown in the next section.
 \section{Downlink $\tt{SINR}$s Model and Achievable Rates Analysis}
 \label{SINR and Rate}
 In this section, we present the analysis of both the $\tt{SINR}$s and the achievable rates for a C-NOMA-based cellular system assisted by RIS for both HD and FD relaying modes. For the two different relaying modes and at each channel use, the superimposed mixture of the signals intended to UE$_{\rm n}$ and UE$_{\rm f}$ at the BS are expressed, respectively, as
 \begin{align}
     \mathcal{S}^{\mathrm{HD}} &= \sqrt{\alpha_{\rm n}^{\mathrm{HD}} P_{\mathrm{BS}}} \mathcal{S}_{\rm n} + \sqrt{\alpha_{\rm f}^{\mathrm{HD}} P_{\mathrm{BS}}} \mathcal{S}_{\rm f}, \cr
     \mathcal{S}^{\mathrm{FD}} &= \sqrt{\alpha_{\rm n}^{\mathrm{FD}} P_{\mathrm{BS}}} \mathcal{S}_{\rm n} + \sqrt{\alpha_{\rm f}^{\mathrm{FD}} P_{\mathrm{BS}}} \mathcal{S}_{\rm f},
 \end{align}
 where $\mathcal{S}_{\rm n}$ and $\mathcal{S}_{\rm f}$ represent the intended signals of UE$_{\rm n}$ and UE$_{\rm f}$, respectively, such that $\mathbb{E}[|\mathcal{S}_{\rm n}|^2] = 1$ and $\mathbb{E}[|\mathcal{S}_{\rm f}|^2] = 1$, $P_{\mathrm{BS}}$ represents the power budget of the BS, and for ${\rm R} \in \{\mathrm{HD}, \mathrm{FD}\}$, $\alpha_{\rm n}^{\rm R}$ and $\alpha_{\rm f}^{\rm R}$ represent the power control coefficients allocated by the BS to UE$_{\rm n}$ and UE$_{\rm f}$ in the relaying mode ${\rm R}$, respectively. Therefore, for ${\rm R} \in \{\mathrm{HD}, \mathrm{FD}\}$, the power allocated by the BS to UE$_{\rm n}$ and UE$_{\rm f}$ in the relaying mode ${\rm R}$ are $\alpha_{\rm n}^{\rm R} P_{\mathrm{BS}}$ and $\alpha_{\rm f}^{\rm R} P_{\mathrm{BS}}$, respectively. 
 \subsection{RIS-enabled HD C-NOMA}
For the case of RIS-enabled HD C-NOMA, the received signal at UE$_{\rm n}$ in the DT phase can be expressed as
\begin{equation}
    \mathcal{Y}_{\rm n}^{\mathrm{HD}} = (h_{\rm bn} + \boldsymbol{h}_{\rm rn}^H \boldsymbol{\Theta}^{\mathrm{HD}}_{(1)} \boldsymbol{h}_{\rm br}) \mathcal{S}^{\mathrm{HD}} + \mathcal{W}_{\rm n},
\end{equation}
where $\mathcal{W}_{\rm n} \sim (0, \sigma_{\rm n}^2)$ is the additive white Gaussian noise (AWGN) with zero mean and variance $\sigma_{\rm n}^2$ and $\boldsymbol{\Theta}^{\mathrm{HD}}_{(1)} = \mathrm{diag}\{e^{j \theta_1^{(1)}}, e^{j \theta_2^{(1)}}, \dots, e^{j \theta_M^{(1)}}\} $ is the phase-shift matrix of the RIS in the DT phase, where for all $m \in \llbracket 1, M \rrbracket$, $\theta_m^{(1)} \in [0, 2\pi]$ represents the phase shift of the $m$th reflecting element of the RIS in the first time slot. According to the NOMA principle, SIC is adopted at UE$_{\rm n}$ to decode
the message $\mathcal{S}_{\rm f}$ of UE$_{\rm f}$. Consequently, the received $\tt{SINR}$ at UE$_{\rm n}$ to decode the message of UE$_{\rm f}$ can be expressed as
\begin{align}
    {\tt{SINR}}_{\rm n \longrightarrow f}^{{\mathrm{HD}}} = \frac{\alpha_{\rm f}^{\mathrm{HD}} P_{\mathrm{BS}}|h_{\rm bn} + \boldsymbol{h}^H_{\rm rn}\boldsymbol{\Theta}_{(1)}^{\mathrm{HD}}\boldsymbol{h}_{\rm br}|^2}{ \alpha_{\rm n}^{\mathrm{HD}} P_{\mathrm{BS}}|h_{\rm bn} + \boldsymbol{h}^H_{\rm rn}\boldsymbol{\Theta}_{(1)}^{\mathrm{HD}}\boldsymbol{h}_{\rm br}|^2 + \sigma_{\rm n}^2}.
\end{align}
\par Therefore, after cancelling UE$_{\rm f}$'s message, the $\tt{SINR}$ at UE$_{\rm n}$ to decode its own message can be expressed as
\begin{align}
    {\tt{SINR}}_{\rm n \longrightarrow n}^{\mathrm{HD}} = \frac{\alpha_{\rm n}^{\mathrm{HD}} P_{\mathrm{BS}}|h_{\rm bn} + \boldsymbol{h}^H_{\rm rn}\boldsymbol{\Theta}_{(1)}^{\mathrm{HD}}\boldsymbol{h}_{\rm br}|^2}{\sigma_{\rm n}^2}.
\end{align}
\par Thus, the achievable data rate at UE$_{\rm n}$ to decode the  message of UE$_{\rm f}$ and to decode its own message can be  expressed, respectively, as
\begin{align}
    \mathcal{R}_{\rm n \longrightarrow f}^{\mathrm{HD}} &= \frac{1}{2}~\mathrm{log}_2 \left(1 + {{\tt{SINR}}_{\rm n \longrightarrow f}^{\mathrm{HD}}}\right), \\
    \mathcal{R}_{\rm n \longrightarrow n}^{\mathrm{HD}} &= \frac{1}{2}~\mathrm{log}_2 \left(1 + {{\tt{SINR}}_{\rm n \longrightarrow n}^{\mathrm{HD}}}\right).
\end{align}
On the other hand, the received signal at UE$_{\rm f}$ in the DT phase is expressed as
\begin{equation}
    \mathcal{Y}_{\rm f}^{(1)} = (h_{\rm bf} + \boldsymbol{h}^H_{\rm rf}\boldsymbol{\Theta}_{(1)}^{\mathrm{HD}}\boldsymbol{h}_{\rm br}) \mathcal{S}^{\mathrm{HD}} + \mathcal{W}_{\rm f}^{(1)},
\end{equation}
where $\mathcal{W}_{\rm f}^{(1)} \sim (0, \sigma_{\rm f}^2)$ is the AWGN at UE$_{\rm f}$ in the first phase. Therefore, the received  $\tt{SINR}$ at UE$_{\rm f}$ in the DT phase is given by
\begin{align}
    {\tt{SINR}}_{(1)}^{\mathrm{HD}} = \frac{\alpha_{\rm f}^{\mathrm{HD}} P_{\mathrm{BS}}|h_{\rm bf} + \boldsymbol{h}^H_{\rm rf}\boldsymbol{\Theta}_{(1)}^{\mathrm{HD}}\boldsymbol{h}_{\rm br}|^2 }{ \alpha_{\rm n}^{\mathrm{HD}} P_{\mathrm{BS}}|h_{\rm bf} + \boldsymbol{h}^H_{\rm rf}\boldsymbol{\Theta}_{(1)}^{\mathrm{HD}}\boldsymbol{h}_{\rm br}|^2 + \sigma_{\rm f}^2}.
\end{align}
\par Meanwhile, in the CT phase, UE$_{\rm n}$ forwards the decoded message $\mathcal{S}_{\rm f}$ to UE$_{\rm f}$. Consequently, the observation at UE$_{\rm f}$ in this phase resulting from the transmission of UE$_{\rm n}$ and the reflections by the RIS can be given by
\begin{equation}
    \mathcal{Y}_{\rm f}^{(2)} = \beta^{\mathrm{HD}} P_{\rm n} (h_{\rm nf} + \hat{\boldsymbol{h}}^H_{\rm rf}\boldsymbol{\Theta}_{(2)}^{\mathrm{HD}}\boldsymbol{h}_{\rm nr}) \mathcal{S}_{\rm f} + \mathcal{W}_{\rm f}^{(2)},
\end{equation}
where $\beta^{\mathrm{HD}} \in [0, 1]$ is the fraction of the allocated power by UE$_{\rm n}$, $P_{\rm n}$ is the power budget at UE$_{\rm n}$, $\hat{\boldsymbol{h}}_{\rm rf}^H$ is the channel gain between the RIS and UE$_{\rm f}$ in the CT phase and $\mathcal{W}_{\rm f}^{(2)} \sim (0, \sigma_{\rm f}^2)$ is the AWGN at UE$_{\rm f}$ in the CT phase. In addition, $\boldsymbol{\Theta}_{(2)}^{\mathrm{HD}}  = \mathrm{diag}\{e^{j \theta_1^{(2)}}, e^{j \theta_2^{(2)}}, \dots, e^{j \theta_M^{(2)}}\}$ is the phase-shift matrix in the CT phase, where for all $m \in \llbracket 1, M \rrbracket$, $\theta_m^{(2)} \in [0, 2\pi]$ represents the phase shift of the $m$th reflecting element of the RIS in the second time-slot. Therefore, the received $\tt{SINR}$ at UE$_{\rm f}$ in the second phase can be expressed as
\begin{equation}
    {\tt{SINR}_{(2)}^{\mathrm{HD}}} = \frac{\beta^{\mathrm{HD}} P_{\rm n} |h_{\rm nf} + \hat{\boldsymbol{h}}^H_{\rm rf}\boldsymbol{\Theta}_{(2)}^{\mathrm{HD}}\boldsymbol{h}_{\rm nr}|^2}{\sigma_{\rm f}^2}.
\end{equation}
\par Toward this end, the effective $\tt{SINR}$ at UE$_{\rm f}$ is the summation of the received $\tt{SINR}$s in the DT and CT phases using maximum ratio combining (MRC) \cite{Dinh_2020}. Based on this, the data rate of UE$_{\rm f}$ due to the MRC can be expressed as
\begin{equation}
    \mathcal{R}_{\mathrm{MRC}}^{\mathrm{HD}} = \frac{1}{2}~\mathrm{log}_{2}( 1 + {\tt{SINR}_{(1)}^{\mathrm{HD}}} + {\tt{SINR}_{(2)}^{\mathrm{HD}}}).
\end{equation}
\par However, the rate $\mathcal{R}_{\mathrm{MRC}}$ can be achievable if and only if UE$_{\rm n}$ has the ability to decode the message $\mathcal{S}_{\rm f}$ of UE$_{\rm f}$. Thus, the data rate achieved at UE$_{\rm f}$ is bounded by the data rate of UE$_{\rm n}$ to decode the message of UE$_{\rm f}$, i.e. $\mathcal{R}_{\rm n \longrightarrow f}$ \cite{Dinh_2020}. Therefore, the achievable data rate of UE$_{\rm f}$ to decode its own message can be given as
\begin{equation}
    \mathcal{R}_{\rm f \longrightarrow f}^{\mathrm{HD}} = \min(\mathcal{R}_{\rm n \longrightarrow f}^{\mathrm{HD}}, \mathcal{R}_{\mathrm{MRC}}^{\mathrm{HD}}). 
\end{equation}
\subsection{RIS-enabled FD C-NOMA}
For the case of RIS-enabled FD C-NOMA, both the DT and CT phases are executed simultaneously using the same radio channel. Due to this, UE$_{\rm n}$ suffers from a SI resulting from receiving data from the BS and transmitting data to the UE$_{\rm f}$ simultaneously within the same co-channel \cite{Liu_2017_Hybrid}. Hence, the received signal at UE$_{\rm n}$ is given by
\begin{equation}
    \mathcal{Y}_{\rm n}^{\mathrm{FD}} =  (h_{\rm bn} + \boldsymbol{h}_{\rm rn}^H \boldsymbol{\Theta}^{\mathrm{FD}} \boldsymbol{h}_{\rm br}) \mathcal{S}^{\mathrm{FD}} + h_{\mathrm{SI}}\tilde{\mathcal{S}}_{\rm f} + \mathcal{W}_{\rm n},
\end{equation}
where $h_{\mathrm{SI}}$ represents the SI channel coefficient at UE$_{\rm n}$ and $\tilde{\mathcal{S}}_{\rm f}$ represents the transmit signal by UE$_{\rm n}$ to UE$_{\rm f}$ \cite{Zhang2017Full}. This transmission of $\tilde{\mathcal{S}}_{\rm f}$ causes an interference at UE$_{\rm n}$. In addition, $\boldsymbol{\Theta}^{\mathrm{FD}}  = \mathrm{diag}\{e^{j \theta_1}, e^{j \theta_2}, \dots, e^{j \theta_M}\}$ is the phase-shift matrix, where for all $m \in \llbracket 1, M \rrbracket$, $\theta_m \in [0, 2\pi]$ represents the phase shift of the $m$th reflecting element of the RIS. Therefore, the achievable data rate of UE$_{\rm n}$ to decode the message of UE$_{\rm f}$ can be expressed as 
\begin{align}
    \mathcal{R}_{\rm n \longrightarrow f}^{{\mathrm{FD}}} = \mathrm{log}_2 \left[ 1 + \frac{\alpha_{\rm f}^{\mathrm{FD}} P_{\mathrm{BS}}|h_{\rm bn} + \boldsymbol{h}^H_{\rm rn}\boldsymbol{\Theta}^{\mathrm{FD}}\boldsymbol{h}_{\rm br}|^2}{ \alpha_{\rm n}^{\mathrm{FD}} P_{\mathrm{BS}}|h_{\rm bn} + \boldsymbol{h}^H_{\rm rn}\boldsymbol{\Theta}^{\mathrm{FD}}\boldsymbol{h}_{\rm br}|^2 + \beta^{\mathrm{FD}} P_{\rm n} \gamma_{\mathrm{SI}} + \sigma_{\rm n}^2}\right], \label{Eq: FD R_n_f}
\end{align}
where $\gamma_{\mathrm{SI}} = |h_{\text{SI}}|^2$. Then, after the successive decoding and cancelling the signal of UE$_{\rm f}$, UE$_{\rm n}$ decodes its own signal $\mathcal{S}_{\rm n}$. Consequently, the achievable data rate of UE$_{\rm n}$ to decode its own signal can be given as 
\begin{align}
    \mathcal{R}^{\mathrm{FD}}_{\rm n \longrightarrow n} = \mathrm{log}_2\left[ 1 + \frac{\alpha_{\rm n}^{\mathrm{FD}} P_{\mathrm{BS}}|h_{\rm bn} + \boldsymbol{h}^{H}_{\rm rn}\boldsymbol{\Theta}^{\mathrm{FD}}\boldsymbol{h}_{\rm br}|^2}{\beta^{\mathrm{FD}} P_{\rm n} \gamma_{\text{SI}} + \sigma_{\rm n}^2}\right].\label{Eq: FD R_n_n}
\end{align}
Afterwards, UE$_{\rm n}$ forwards the signal $\tilde{\mathcal{S}}_{\rm f}$ to UE$_{\rm f}$. Thus, the received signal at UE$_{\rm f}$ is given by
\begin{equation}
    \mathcal{Y}^{\mathrm{FD}} = (h_{\rm bf} + \boldsymbol{h}^H_{\rm rf}\boldsymbol{\Theta}^{\mathrm{FD}}\boldsymbol{h}_{\rm br})\mathcal{S}^{\mathrm{FD}} + \beta^{\mathrm{FD}} P_{\rm n} (h_{\rm nf} + \boldsymbol{h}^H_{\rm rf}\boldsymbol{\Theta}^{\mathrm{FD}}\boldsymbol{h}_{\rm nr}) \tilde{\mathcal{S}}_{\rm f} + \mathcal{W}_{\rm f}, \label{FD_Y_F}
\end{equation}
where the first term in \eqref{FD_Y_F} is due to the transmission of the BS, whereas the second term results from the transmission of the near user over the D2D communication link. Note that, in FD scenario, $\boldsymbol{h}_{\rm nr} \triangleq \boldsymbol{h}_{\rm rn}$. Moreover,  UE$_{\rm f}$ receives its own message from the transmission of both BS and UE$_{\rm n}$ at approximately the same channel use \cite{Elhattab2020AJoint, Wei2018Energy, Zhang2019Resource}. Hence, according to \cite{Elhattab2020AJoint, Wei2018Energy, Zhang2019Resource}, UE$_{\rm f}$ can successfully align, co-phase, and combine the signal $\mathcal{S}_{\rm f}$ transmitted from the BS and the signal $\tilde{\mathcal{S}}_{\rm f}$ (which is in fact the signal $\mathcal{S}_{\rm f}$ of UE$_{\rm f}$ decoded by UE$_{\rm n}$) forwarded from UE$_{\rm n}$. Consequently, based on the above discussion and on the results of \cite{Elhattab2020AJoint, Wei2018Energy, Zhang2019Resource}, the achievable data rate of UE$_{\rm f}$ can be expressed as
\begin{align}
    \mathcal{R}^{\mathrm{FD}}_{\mathrm{MRC}} = \mathrm{log}_2 \left[ 1 + \frac{\alpha_{\rm f}^{\mathrm{FD}} P_{\mathrm{BS}}|h_{\rm bf} + \boldsymbol{h}^H_{\rm rf}\boldsymbol{\Theta}^{\mathrm{FD}}\boldsymbol{h}_{\rm br}|^2 + \beta^{\mathrm{FD}} P_{\rm n} |h_{\rm nf} + \boldsymbol{h}^H_{\rm rf}\boldsymbol{\Theta}^{\mathrm{FD}}\boldsymbol{h}_{\rm nr}|^2 }{ \alpha_{\rm n}^{\mathrm{FD}} P_{\mathrm{BS}}|h_{\rm bf} + \boldsymbol{h}^H_{\rm rf}\boldsymbol{\Theta}^{\mathrm{FD}}\boldsymbol{h}_{\rm br}|^2 +  \sigma^2_{\rm f}}\right]. \label{Eq: FD R_MRC} 
\end{align}
\par Based on this analysis and according to what was explained in the HD case, the achievable data rate of UE$_{\rm f}$ to decode its own message can be given by
\begin{equation}
    \mathcal{R}_{\rm f \longrightarrow f}^{\mathrm{FD}} = \min(\mathcal{R}_{\rm \rm n \longrightarrow f}^{\mathrm{FD}}, \mathcal{R}^{\mathrm{FD}}_{\mathrm{MRC}}). 
\end{equation}
\section{RIS-Enabled HD C-NOMA: Problem Formulation and Solution Approach}
\label{Sec: HD-CNOMA RIS}
\subsection{Problem Formulation}
With the quest of improving the performance of the proposed RIS-enabled HD C-NOMA, an optimization problem is formulated with the objective of minimizing the total transmit power by the BS and the near user. This proposed framework includes two different objectives: 1) Power control for both the BS and the near user and 2) Phase-shift for the RIS, i.e., finding the best RIS configuration in the first and second time slots to enhance the system performance. By optimizing the power allocation coefficients at the BS $(\alpha^{\mathrm{HD}}_{\rm n}, \alpha^{\mathrm{HD}}_{\rm f})$, the power fraction coefficient at the near user $\beta^{\mathrm{HD}}$, and the phase-shift matrices for the RIS $\boldsymbol{\theta}_{(1)} = [\theta_1^{(1)}, \theta_2^{(1)}, \dots, \theta_M^{(1)}]$ and $\boldsymbol{\theta}_{(2)} = [\theta_1^{(2)}, \theta_2^{(2)}, \dots, \theta_M^{(2)}]$, the total transmit power minimization problem for the proposed RIS-enabled HD C-NOMA framework can be formulated as follows.
\allowdisplaybreaks
\begingroup
\begin{subequations}
\label{prob:ROS HD C_NOMA}
\begin{align}
&\mathrm{OPT}-\mathrm{HD}: \min_{ \substack{\boldsymbol{\theta}_{(1)}, \boldsymbol{\theta}_{(2)}, \\ \alpha_{\rm n}^{\mathrm{HD}},\alpha_{\rm f}^{\mathrm{HD}}, \beta^{\mathrm{HD}}}}   (\alpha_{\rm n}^{\mathrm{HD}} + \alpha_{\rm f}^{\mathrm{HD}}) P_{\mathrm{BS}} + \beta^{\mathrm{HD}} P_{\rm n}, \\
&\text{s.t.}\,\, 0 \leq \alpha_{\rm n}^{\mathrm{HD}} \leq \alpha_{\rm f}^{\mathrm{HD}}, \label{P1_C1}\\
&\quad \,\,\, 0 \leq \alpha_{\rm n}^{\mathrm{HD}} + \alpha_{\rm f}^{\mathrm{HD}}  \leq 1, \label{P1_C2}\\
&\quad \,\,\, 0 \leq \beta^{\mathrm{HD}} \leq 1, \label{P1_C3}\\
&\quad \,\,\, \mathcal{R}_{\rm n \longrightarrow n}^{\mathrm{HD}} \geq  R^{\rm th}_{\rm n}, \label{P1_C4}\\
&\quad \,\,\, \mathcal{R}_{\mathrm{MRC}}^{\mathrm{HD}} \geq R_{\rm f}^{\rm th}, \label{P1_C5}\\ 
&\quad \,\,\, \mathcal{R}_{\rm n \longrightarrow f}^{\mathrm{HD}} \geq R_{\rm f}^{\rm th}, , \label{P1_C6}\\
&\quad \,\,\, 0\leq \theta _{m}^{(1)} \leq 2\pi, \quad \forall\,\, m \in \llbracket 1, M \rrbracket \label{P1_C7}\\
&\quad \,\,\, 0\leq \theta _{m}^{(2)} \leq 2\pi, \quad \forall\,\, m \in \llbracket 1, M \rrbracket \label{P1_C8}
\end{align}
\end{subequations}
\endgroup
where constraint \eqref{P1_C1} represents the SIC constraint, and constraints \eqref{P1_C2} and \eqref{P1_C3} guarantee that the total transmit power by the BS and UE$_{\rm n}$ do not exceed their power budget, respectively. Constraints \eqref{P1_C4}-\eqref{P1_C6} represent the QoS constraints for UE$_{\rm n}$ and UE$_{\rm f}$, where $R_{\rm n}^{\mathrm{th}}$ and $R_{\rm f}^{\mathrm{th}}$ represent the required QoS in terms of minimum data rate for UE$_{\rm n}$ and UE$_{\rm f}$, respectively. It can be seen that it is challenging to solve $\mathrm{OPT-HD}$ directly due to the high coupling between the power allocation coefficients ($\alpha_{\rm n}^{\mathrm{HD}}, \alpha_{\rm f}^{\mathrm{HD}}, \beta^{\mathrm{HD}}$) and the RIS phase-shift coefficients ($\boldsymbol{\theta}_{(1)}, \boldsymbol{\theta}_{(2)}$) as well as the non convexity of constraints \eqref{P1_C4}-\eqref{P1_C6}. Therefore, problem $\mathrm{OPT-HD}$ is hard to be solved by common standard optimization techniques. Thus, it is necessary to transform problem $\mathrm{OPT-HD}$ into tractable sub-problems that can be solved alternatively. Toward this end, the alternating optimization approach is utilized to solve $\mathrm{OPT-HD}$ in an efficient manner.
\par Accordingly, unlike \cite{zuo2020reconfigurable}, which decomposes the main problem into the DT sub-problem and CT sub-problem, we divide problem $\mathrm{OPT-HD}$ into two sub-problems, i.e. power control optimization sub-problem and RIS passive beamforming (phase-shift coefficients) optimization sub-problem. This is because we seek for deriving optimal closed-form expressions for the power control coefficients $(\alpha_{\rm n}^{\mathrm{HD}}, \alpha_{\rm f}^{\mathrm{HD}}, \beta^{\mathrm{FD}})$ as a function of the phase-shift coefficients ($\boldsymbol{\theta}_{(1)}, \boldsymbol{\theta}_{(2)}$). In particular, we start by optimizing the power allocation coefficients at the BS and the power fraction coefficient at the near UE given the phase-shift coefficients $\boldsymbol{\theta}_{(1)}$ and $\boldsymbol{\theta}_{(2)}$ at the RIS in the DT and CT phases. Then, once the optimal power allocation coefficients at the BS and at UE$_{\rm n}$ are obtained, we direct our attention to the optimal phase-shift coefficients of the RIS at the DT and CT phases. Based on the above discussion, the power control optimization problem can be written as
\allowdisplaybreaks
\begingroup
\begin{subequations}
\label{prob: Power Control}
\begin{align}
&\mathrm{PC-HD}: \min_{ \alpha_{\rm n}^{\mathrm{HD}},\alpha_{\rm f}^{\mathrm{HD}}, \beta^{\mathrm{HD}}}   (\alpha_{\rm n}^{\mathrm{HD}} + \alpha_{\rm f}^{\mathrm{HD}}) P_{\mathrm{BS}} + \beta^{\mathrm{HD}} P_{\rm n}, \\
&\text{s.t.}\,\, \eqref{P1_C1}-\eqref{P1_C6},
\end{align}
\end{subequations}
\endgroup
whereas the passive beamforming optimization problem can be presented as
\allowdisplaybreaks
\begingroup
\begin{subequations}
\label{prob: Phase-shift Opt}
\begin{align}
&\mathrm{PS-HD}: \mathrm{Find}  ~~ \boldsymbol{\theta}_{(1)}, \boldsymbol{\theta}_{(2)}  \\
&\text{s.t.}\,\, \eqref{P1_C4}-\eqref{P1_C8}
\end{align}
\end{subequations}
\endgroup
\subsection{RIS-enabled HD C-NOMA: Power Control Optimization}
\label{sec:PC-HD}
In this part, we assume that the phase shift matrices $\boldsymbol{\Theta}_{(1)}^{\mathrm{HD}}$ and $\boldsymbol{\Theta}_{(2)}^{\mathrm{HD}}$ are fixed. Based on this, we denote by 
\begin{equation}
    \gamma_{\rm bn} \triangleq \frac{P_{\mathrm{BS}}|h_{\rm bn} + \boldsymbol{h}^H_{\rm rn}\boldsymbol{\Theta}_{(1)}^{\mathrm{HD}}\boldsymbol{h}_{\rm br}|^2}{\sigma_{\rm n}^2},\,\,
    \gamma_{\rm bf} \triangleq \frac{P_{\mathrm{BS}}|h_{\rm bf} + \boldsymbol{h}^H_{\rm rf}\boldsymbol{\Theta}_{(1)}^{\mathrm{HD}}\boldsymbol{h}_{\rm br}|^2 }{\sigma_{\rm f}^2}, \,\, \text{and}\,\,
    \gamma_{\rm d} \triangleq \frac{P_{\rm n} |h_{\rm nf} + \hat{\boldsymbol{h}}^H_{\rm rf}\boldsymbol{\Theta}_{(2)}^{\mathrm{HD}}\boldsymbol{h}_{\rm nr}|^2}{\sigma_{\rm f}^2}.
\end{equation}
Before deriving the optimal power control of problem $\mathrm{PC-HD}$, one needs to specify its feasibility conditions. The  feasibility  conditions  of  problem $\mathrm{PC-HD}$ define  the conditions under  which at least  one feasible solution  for this problem does exist. In addition, a feasible solution for problem $\mathrm{PC-HD}$ defines a solution  that satisfies the constraint of this problem. In this context, the feasibility conditions of problem $\mathrm{PC-HD}$ are presented in the following theorem.
\begin{theorem}
Problem $\mathrm{PC-HD}$ is feasible if and only if the following conditions hold.
\begin{subequations}
\label{eq:cdts_HD}
\begin{align}
    &\text{Condition 1:} \,\, \alpha_{\min}^{\rm HD} \leq \alpha_{\max}^{\rm HD}, \label{eq:cdt1}\\
    &\text{Condition 2:} \,\, \beta_{\min}^{\rm HD} \leq \beta_{\max}^{\rm HD}, \label{eq:cdt2}
\end{align}
\end{subequations}
where $\alpha_{\min}^{\rm HD}$, $\alpha_{\max}^{\rm HD}$, $\beta_{\min}^{\rm HD}$ and $\beta_{\max}^{\rm HD}$ are expressed, respectively, as
\begin{equation}
    \label{HD_bounds}
    \begin{aligned}
    &\alpha_{\min}^{\rm HD} = \frac{t_{\rm n}^{\rm HD}}{\gamma_{\rm bn}},  &\alpha_{\max}^{\rm HD} = \min \left(0, \frac{\gamma_{\rm bn} - t_{\rm f}^{\rm HD}}{\gamma_{\rm bn} \left(t_{\rm f}^{\rm HD}+1\right)} \right), \\ 
    &\beta_{\min}^{\rm HD} = \max \left(0, \frac{1}{\gamma_{\rm d}} \left(t_{\rm f}^{\rm HD} - \frac{\gamma_{\rm bn} - t_{\rm n}^{\rm HD}}{t_{\rm n}^{\rm HD} + \frac{\gamma_{\rm bn}}{\gamma_{\rm bf}}} \right) \right), &\beta_{\max}^{\rm HD} = 1,
    \end{aligned}
\end{equation}
such that $t_{\rm n}^{\rm HD} = 2^{2 R_{\rm n}^{\rm th}} - 1$ and $t_{\rm f}^{\rm HD} = 2^{2 R_{\rm n}^{\rm th}} - 1$.
\end{theorem}
\begin{proof}
See Appendix A.
\end{proof}
Afterwards, assuming that problem $\mathrm{PC-HD}$ is feasible, its optimal solution is given in the following theorem. 
\allowdisplaybreaks
\begingroup
\begin{theorem}
Assuming that problem $\mathrm{PC-HD}$ is feasible, i.e., conditions \eqref{eq:cdt1} and \eqref{eq:cdt2} hold, its optimal solution is expressed as follows. Let $\mathbf{p}_1^{\rm HD}$ and $\mathbf{p}_2^{\rm HD}$ denote the power control scheme expressed, respectively, as
\begin{subequations}
\begin{align}
    &\mathbf{p}_1^{\rm HD} = \left(\alpha_{\rm n,1}^{\rm HD},\alpha_{\rm f, 1}^{\rm HD}, \beta_{1}^{\rm HD} \right) = \left(\alpha_{\min}^{\rm HD}, \max\left(\alpha_{\min}^{\rm HD},\alpha_{\min}^{\rm HD} t_{\rm f}^{\rm HD} + \frac{t_{\rm f}^{\rm HD}}{\gamma_{\rm n}}\right),\frac{1}{\gamma_{\rm d}} \left(t_{\rm f}^{\rm HD} - \frac{\alpha_{\rm f,1}^{\rm HD} \gamma_{\rm f}}{\alpha_{\rm n,1}^{\rm HD}\gamma_{\rm f}+1} \right) \right),\\
    &\mathbf{p}_2^{\rm HD} = \left(\alpha_{\rm n, 2}^{\rm HD},\alpha_{\rm f, 2}^{\rm HD}, \beta_{2}^{\rm HD} \right) = \left(\alpha_{\min}^{\rm HD}, \frac{\left(\alpha_{\min}^{\rm HD}\gamma_{\rm bn} + 1\right)t_{\rm f}^{\rm HD}}{\gamma_{\rm bn}},0 \right).
\end{align}
\end{subequations}
Then, the optimal power control scheme is expressed as
\begin{equation}
    \mathbf{p}^{{\rm HD}^*} = \left(\alpha_{\rm n}^{{\rm HD}^*}, \alpha_{\rm f}^{{\rm HD}^*},\beta^{{\rm HD}^*} \right) = \argmin_{\mathbf{p}^{\rm HD} \in \left\{\mathbf{p}_1^{\rm HD}, \mathbf{p}_2^{\rm HD}\right\}}  f \left(\mathbf{p}^{\rm HD} \right),
\end{equation}
where the function $f$ is the objective function of problem $\mathrm{PC-HD}$, i.e., 
\begin{equation}
    f \left(\mathbf{p}^{\rm HD} \right) = f\left(\alpha_{\rm n}^{\mathrm{HD}},\alpha_{\rm f}^{\mathrm{HD}}, \beta^{\mathrm{HD}}\right) = (\alpha_{\rm n}^{\mathrm{HD}} + \alpha_{\rm f}^{\mathrm{HD}}) P_{\mathrm{BS}} + \beta^{\mathrm{HD}} P_{\rm n}.
\end{equation}
\end{theorem}
\endgroup
\begin{proof}
See Appendix B.
\end{proof}
\subsection{RIS-enabled HD C-NOMA: Phase-Shift Coefficients Optimization}
In this subsection, the phase-shift coefficients in both DT and CT phases are optimized with given values of $(\alpha_{\rm n}^{\mathrm{HD}}, \alpha_{\rm f}^{\mathrm{HD}}, \beta^{\mathrm{HD}})$. One can see that $\mathrm{PS-HD}$ is a feasibility check problem (finding the phase-shift coefficient for each element such that the QoS constraints are satisfied). Note that, the RIS provides additional paths to construct a stronger combined channel gain at the intended receiver. Therefore, the best channel gain in the second time-slot can be achieved when the reflected signals from the RIS can be constructively added at UE$_{\rm f}$ and be co-phased with the direct D2D link from UE$_{\rm n}$ to UE$_{\rm f}$. As a result, optimal phase-shift coefficients in CT phase can be obtained as follows \cite{Elhattab_2020, Bjornson_2020_Intelligent}.
\begin{equation}
    \theta_m^{(2)} = \mathrm{arg}(h_{\rm nf}) - \mathrm{arg}([\boldsymbol{h}_{\rm nr}]_m [\hat{\boldsymbol{h}}_{\rm rf}]_m). 
\end{equation}
Consequently, $\tt{SINR}_{(2)}^{\mathrm{HD}}$ can be expressed as
\begin{equation}
    {\tt{SINR}}_{(2)}^{\mathrm{HD}} = \frac{\beta^{\mathrm{HD}} P_{\rm n} (|h_{\rm nf}| + \sum_{m = 1}^{M} |[\hat{\boldsymbol{h}}_{\rm rf}]_m [\boldsymbol{h}_{\rm nr}]_m|)^2}{\sigma_{\rm f}^2}.
\end{equation}
After obtaining the optimal value for $\theta_{(2)}^{\mathrm{HD}}$, the SDR technique is applied for obtaining $\theta_{(1)}^{\mathrm{HD}}$. Let us start by defining $\boldsymbol {v} = [v_{1}, \dots, v_{M}]^{H}$, where for all $m \in \llbracket 1, M \rrbracket$, $v_{m} = e^{j\theta _{m}}$. Then, the constraints in \eqref{P1_C7} are equivalent to the unit-modulus constraints, i.e., $|v_{m}|^{2}=1$ for all $m \in \llbracket 1, M \rrbracket$. By applying the change of variables $\boldsymbol {h}^{H}_{\rm rn} \boldsymbol{\Theta}_{(1)}^{\mathrm{HD}} \boldsymbol{h}_{\rm br} = \boldsymbol {v}^{H}\boldsymbol {\Phi}$, where $\boldsymbol{\Phi}= \mathrm{diag}(\boldsymbol {h}^{H}_{\rm rn})\boldsymbol {h}_{\rm br} \in \mathbb {C}^{M \times 1}$ and $\boldsymbol {h}^{H}_{\rm rf} \boldsymbol{\Theta} \boldsymbol {h}_{\rm br} = \boldsymbol {v}^{H}\boldsymbol {\Psi}$, where $\boldsymbol {\Psi }=\text {diag}(\boldsymbol {h}^{H}_{\rm rf})\boldsymbol {h}_{\rm br} \in \mathbb {C}^{M \times 1}$, we have
\begin{equation}
    |h_{\rm bn} + \boldsymbol {h}^{H}_{\rm rn} \boldsymbol{\Theta} \boldsymbol {h}_{\rm br}|^{2} =|h_{\rm bn} + \boldsymbol {v}^{H}\boldsymbol {\Phi}|^{2} \quad \text{and} \quad
    |h_{\rm bf} + \boldsymbol {h}^{H}_{\rm rf} \boldsymbol{\Theta} \boldsymbol {h}_{\rm br}|^{2} =|h_{\rm bf} + \boldsymbol {v}^{H}\boldsymbol {\Psi}|^{2}. 
\end{equation}
By introducing an auxiliary variable $t$, then an equivalent representation of the achievable data rates $\mathcal{R}^{\mathrm{HD}}_{\rm n \longrightarrow n}, \mathcal{R}^{\mathrm{HD}}_{\mathrm{MRC}}, $ and $\mathcal{R}^{\mathrm{HD}}_{\rm n \longrightarrow f}$ can be expressed as follows 
\begin{align}
    \mathcal{R}^{\mathrm{HD}}_{\rm n \longrightarrow n} = & \mathrm{log}\left( 1 + \frac{\alpha_{\rm n}^{\mathrm{HD}} P_{\mathrm{BS}} (\bar{\boldsymbol {v}}^{H}\boldsymbol {Q}_{\rm bn}\bar{\boldsymbol {v}} + | {h}_{\rm bn}|^{2})}{\sigma_{\rm n}^2}\right), \\
    \mathcal{R}^{\mathrm{HD}}_{\mathrm{MRC}} = & \mathrm{log}\left(1 + \frac{\alpha_{\rm f}^{\mathrm{HD}} P_{\mathrm{BS}} (\bar{\boldsymbol {v}}^{H}\boldsymbol {Q}_{\rm bf}\bar{\boldsymbol {v}} + |{h}_{\rm bf}|^{2}) }{\alpha_{\rm n}^{\mathrm{HD}} P_{\mathrm{BS}} (\bar{\boldsymbol {v}}^{H}\boldsymbol {Q}_{\rm bf}\bar{\boldsymbol {v}} + |{h}_{\rm bf}|^{2}) + \sigma_{\rm f}^2} + {\tt{SINR}}_{(2)}^{\mathrm{HD}}\right), \\  \mathcal{R}^{\mathrm{HD}}_{\rm n \longrightarrow f} = & \mathrm{log}\left( 1 + \frac{\alpha^{\mathrm{HD}}_{\rm f} P_{\mathrm{BS}} (\bar{\boldsymbol {v}}^{H}\boldsymbol {Q}_{\rm bn}\bar{\boldsymbol {v}} + | {h}_{\rm bn}|^{2})}{\alpha^{\mathrm{HD}}_{\rm n} P_{\mathrm{BS}} (\bar{\boldsymbol {v}}^{H}\boldsymbol {Q}_{\rm bn}\bar{\boldsymbol {v}} + | {h}_{\rm bn}|^{2}) + \sigma_{\rm n}^2}\right),
\end{align}
where 
\begin{equation} \boldsymbol {Q}_{\rm bn}=\begin{bmatrix} \boldsymbol {\Phi }\boldsymbol {\Phi }^{H} &\boldsymbol {\Phi } {h}^H_{\rm bn} \\ {h}_{\rm bn}\boldsymbol {\Phi }^{H} &0 \\ \end{bmatrix}, \quad \boldsymbol {Q}_{\rm bf}=\begin{bmatrix} \boldsymbol {\Psi }\boldsymbol {\Psi }^{H} &\boldsymbol {\Psi } {h}^H_{\rm bf} \\ {h}_{\rm bf}\boldsymbol {\Psi }^{H} & 0 \\ \end{bmatrix} \quad \text{and}  \quad \bar{\boldsymbol {v}}=\begin{bmatrix} \boldsymbol {v} \\ t \\ \end{bmatrix}.\end{equation} 
\par Note that $\bar{\boldsymbol {v}}^{H}\boldsymbol {Q}_{z}\bar{\boldsymbol {v}}={\mathrm{tr}}(\boldsymbol {Q}_{z}\bar{\boldsymbol {v}}\bar{\boldsymbol {v}}^{H})$ for all $z \in \{{\rm bn}, \rm bf\}$. In addition, define $\boldsymbol {V}=\bar{\boldsymbol {v}}\bar{\boldsymbol {v}}^{H}$, which needs to satisfy ${\mathrm{rank}}(\boldsymbol {V}) = 1$ and $\boldsymbol {V}\succeq \boldsymbol {0}$. This rank-one constraint (${\mathrm{rank}}(\boldsymbol {V})=1$) is non-convex \cite{Zhang_TWC_2019}. Consequently, by dropping this constraint, $\mathrm{PS-HD}$ can be rewritten as 
\begingroup
\begin{subequations}
\label{prob:SDR-HD}
\begin{align}
&\mathcal{P}: \mathrm{Find}\quad \boldsymbol{\theta}_{(1)}^{\mathrm{HD}} \\
&\text{s.t.}\,\,\alpha_{\rm n}^{\mathrm{HD}}P_{\mathrm{BS}} ({\mathrm{tr}}(\boldsymbol {Q}_{\rm bn}\boldsymbol{V}) + |h_{\rm bn}|^2) \geq t_{\rm n}^{\rm HD}\sigma_{\rm n}^2, \label{HD:P2_C1}\\
&\quad \,\,\,  \alpha_{\rm f}^{\mathrm{HD}}P_{\mathrm{BS}}({\mathrm{tr}}(\boldsymbol {Q}_{\rm bf}\boldsymbol{V}) + |h_{\rm bf}|^2) \geq (t_{\rm f}^{\rm HD} - {\tt{SINR}}_{(2)}^{\mathrm{HD}})(\alpha_{\rm n}^{\mathrm{HD}}P_{\mathrm{BS}}({\mathrm{tr}}(\boldsymbol{Q}_{\rm bf}\boldsymbol{V}) + |h_{\rm bf}|^2) + \sigma_{\rm f}^2), \label{HD:P2_C2}\\
&\quad \,\,\,\alpha_{\rm f}^{\mathrm{HD}}P_{\mathrm{BS}} ({\mathrm{tr}}(\boldsymbol {Q}_{\rm bn}\boldsymbol{V}) + |h_{\rm bn}|^2) \geq t_{\rm f}^{\rm HD}(\alpha_{\rm n}^{\mathrm{HD}}P_{\mathrm{BS}} ({\mathrm{tr}}(\boldsymbol {Q}_{\rm bn}\boldsymbol{V}) + |h_{\rm bn}|^2) + \sigma_{\rm n}^2), \label{HD:P2_C4} \\&\quad\,\,\, \boldsymbol {V} \succeq 0, \\ &\quad\,\,\, [\boldsymbol{V}]_{m,m} = 1, \qquad\qquad \forall\,\, m \in \llbracket 1, M+1 \rrbracket.
\end{align}
\end{subequations}
\endgroup
It is not difficult to observe that problem $\mathcal{P}$ is a semi-definite programming (SDP) problem and hence it can be optimally solved by existing convex optimization solvers such as CVX \cite{Zhang_TWC_2019}. In general, the optimal $\boldsymbol{V}$ obtained by solving problem $\mathcal{P}$ does not satisfy the rank-one constraint. This implies that the optimal solution of problem $\mathcal{P}$ only serves as an upper bound for $\mathrm{PS-HD}$. Consequently, additional steps are required to construct a rank-one solution, which can be achieved by applying the Gaussian randomization scheme to obtain a rank-one solution. This can be described as follows. First, we obtain the eigenvalue decomposition of $\boldsymbol{V}$ as $\boldsymbol{V} = \boldsymbol{U} \boldsymbol{\Sigma} \boldsymbol{U}^H$, where $\boldsymbol{U} = [u_1, u_2, \dots, u_{M + 1}]$ is a unitary matrix and $\boldsymbol{\Sigma} = \mathrm{diag}(\lambda_1, \lambda_2, \dots, \lambda_{M + 1})$ is a diagonal matrix, respectively. After that, a random vector is generated as $\bar{v} = \boldsymbol{U} \boldsymbol{\Sigma}^{1/2}\mathbf{r}$, where $\mathbf{r}$ is a random vector that follows
a circularly symmetric complex Gaussian (CSCG) distribution with a zero mean and a co-variance matrix equal to the identity matrix of order $M+1$, denoted by $\boldsymbol{I}_{M + 1}$, i.e., $\mathbf{r} \sim \mathcal{CN}(\boldsymbol{0}, \boldsymbol{I}_{M + 1})$. Furthermore, we generate the scalar $\boldsymbol{v} = \exp\left[{j~\mathrm{arg}\left(\frac{\left[\bar{\boldsymbol{v}}\right]_{1:M}}{\left[\bar{\boldsymbol{v}}\right]_{M + 1}}\right)}\right]$, where $[\boldsymbol{x}]_{1:M}$ denotes a vector having the first $M$ elements in $\boldsymbol{x}$. It is important to mention that the SDR approach followed by a large number of Gauss randomization can guarantee a minimum accuracy of $\pi/4$ of the optimal objective value \cite{Zhang_TWC_2019}. Finally, the steps of the proposed optimization scheme for RIS-enabled HD C-NOMA are presented in details in \textbf{Algorithm 1}. The convergence of the alternating optimization technique is guaranteed and the proof details can be found in \cite{bezdek2003convergence}.
\begin{algorithm}[!t]
\DontPrintSemicolon
\small{
\caption{\small{Alternating Optimization Algorithm for RIS-enabled HD C-NOMA}}
\KwIn{$P_{\mathrm{BS}}, P_{\rm n}, \boldsymbol{Q}_{\rm bn}, \boldsymbol{Q}_{\rm bf}, \sigma^2_{\rm n}, \sigma_{\rm f}^2$ and ${\tt{SINR}}_{(2)}^{\mathrm{HD}}$\;}
Initialize the phase-shift $\boldsymbol{\theta}_{(1)}^{\mathrm{HD}, 0}$ and set the iteration number $l = 1$\;
\textbf{repeat}\;
Using the closed-form expressions in \textit{Theorem} 2 to find the optimal power allocation coefficients at the BS $(\alpha_{\rm n}^{\mathrm{HD}, l}, \alpha_{\rm f}^{\mathrm{HD}, l})$ and the optimal power fraction coefficient at UE$_{\rm n}$, i.e. $\beta^{\mathrm{HD}, l}$ for given $\boldsymbol{\theta}_{(1)}^{\mathrm{HD}, l}$  \; 
Solve problem $\mathcal{P}$ for given power allocation coefficients $(\alpha_{\rm n}^{\mathrm{HD}, l}, \alpha_{\rm f}^{\mathrm{HD}, l}, \beta^{\mathrm{HD}, l})$\;
Obtain an approximate solution for the phase-shift using eigenvalue decomposition and Gauss randomization approach. Then, denote the solution as $\boldsymbol{\theta}_{(1)}^{\mathrm{HD}, l}$\;
Update $l = l + 1$\;
\textbf{Until} The decrease of the objective value in \eqref{prob:ROS HD C_NOMA} is below a threshold $\epsilon > 0$  or the maximum number of iterations $L$ is reached
}
\end{algorithm}
\section{RIS-Enabled FD C-NOMA: Problem Formulation and Solution Approach}
\label{Sec:FD C-NOMA RIS}
\subsection{Problem Formulation}
In this subsection, we investigate the  power minimization problem for RIS-enabled FD C-NOMA systems. In the FD C-NOMA, UE$_{\rm n}$ decodes the message of UE$_{\rm f}$ and then forwards it to UE$_{\rm f}$ through the D2D direct link with the aid of the RIS in the same time-slot. As a result, in contrast to the HD scenario that requires two time-slot and adjusts the RIS's configuration in each one of them, only one RIS's configuration is required in the FD case. By optimizing the power allocation coefficients at the BS $(\alpha^{\mathrm{FD}}_{\rm n}, \alpha^{\mathrm{FD}}_{\rm f})$, the power fraction coefficient at UE$_{\rm n}$, $\beta^{\mathrm{FD}}$, and the phase-shift coefficients for the RIS $\boldsymbol{\theta}^{\mathrm{FD}} = [\theta_1, \theta_2, \dots, \theta_M]$, the total transmit power minimization problem for the proposed RIS-enabled FD C-NOMA framework can be formulated as follows.
\allowdisplaybreaks
\begingroup
\begin{subequations}
\label{prob:RIS FD C_NOMA}
\begin{align}
&\mathrm{OPT}-\mathrm{FD}: \min_{ \substack{\boldsymbol{\theta}^{\mathrm{FD}}, \alpha_{\rm n}^{\mathrm{FD}},\\ \alpha_{\rm f}^{\mathrm{FD}}, \beta^{\mathrm{FD}}}}   (\alpha^{\mathrm{FD}}_{\rm n} + \alpha^{\mathrm{FD}}_{\rm f}) P_{\mathrm{BS}} + \beta^{\mathrm{FD}} P_{\rm n}, \\
&\text{s.t.}\,\, 0 \leq \alpha_{\rm n}^{\mathrm{FD}} \leq \alpha_{\rm f}^{\mathrm{FD}}, \label{P3_C1}\\
&\quad \,\,\, 0 \leq \alpha_{\rm n}^{\mathrm{FD}} + \alpha_{\rm f}^{\mathrm{FD}}  \leq 1, \label{P3_C2}\\
&\quad \,\,\, 0 \leq \beta^{\mathrm{FD}} \leq 1, \label{P3_C3}\\
&\quad \,\,\, \mathcal{R}_{\rm n \longrightarrow n}^{\mathrm{FD}} \geq  R^{\rm th}_{\rm n}, \label{P3_C4}\\
&\quad \,\,\, \mathcal{R}_{\mathrm{MRC}}^{\mathrm{FD}} \geq R_{\rm f}^{\rm th}, \label{P3_C5}\\ 
&\quad \,\,\, \mathcal{R}_{\rm n \longrightarrow f}^{\mathrm{FD}} \geq R_{\rm f}^{\rm th}, \label{P3_C6}\\
&\quad \,\,\, 0\leq \theta_{m} \leq 2\pi, \quad \forall\,\, m \in \llbracket 1, M \rrbracket \label{P3_C7}
\end{align}
\end{subequations}
\endgroup
Similar to $\mathrm{OPT-HD}$, $\mathrm{OPT-FD}$ is hard to be solved by common standard optimization techniques. Therefore, we resort to the alternating optimization technique similar to the RIS-enabled HD C-NOMA case in solving problem $\mathrm{OPT-FD}$. Consequently, the power control optimization problem with a fixed $\boldsymbol{\theta}^{\mathrm{FD}}$ can be written as
\allowdisplaybreaks
\begingroup
\begin{subequations}
\label{prob: FD Power Control}
\begin{align}
&\mathrm{PC-FD}: \min_{ \alpha_{\rm n}^{\mathrm{FD}},\alpha_{\rm f}^{\mathrm{FD}}, \beta^{\mathrm{FD}}}   (\alpha_{\rm n}^{\mathrm{FD}} + \alpha_{\rm f}^{\mathrm{FD}}) P_{\mathrm{BS}} + \beta^{\mathrm{FD}} P_{\rm n}, \\
&\text{s.t.}\,\, \eqref{P3_C1}-\eqref{P3_C6},
\end{align}
\end{subequations}
\endgroup
whereas the passive beamforming optimization problem with given power allocation coefficients, i.e., $\alpha_{\rm n}^{\mathrm{FD}}, \alpha_{\rm f}^{\mathrm{FD}}$ and $\beta^{\mathrm{FD}}$ can be presented as
\allowdisplaybreaks
\begingroup
\begin{subequations}
\label{prob: FD Phase-shift Opt}
\begin{align}
&\mathrm{PS-FD}: \mathrm{Find}  ~~ \boldsymbol{\theta}^{\mathrm{FD}},  \\
&\text{s.t.}\,\, \eqref{P3_C4}-\eqref{P3_C7}.
\end{align}
\end{subequations}
\endgroup 
\subsection{RIS-enabled FD C-NOMA: Power Allocation Optimization}
In this part, we start by determining the feasibility conditions of problem $\mathrm{PC-FD}$. Assuming that the phase shift matrix $\boldsymbol{\Theta}^{\mathrm{FD}}$ is fixed, we denote by $\gamma_{\rm SI} \triangleq \frac{P_{\rm n} \gamma_{\text{SI}}}{\sigma_{\rm n}^2}$ and by
\begin{equation}
    \gamma_{\rm bn} \triangleq \frac{P_{\mathrm{BS}}|h_{\rm bn} + \boldsymbol{h}^H_{\rm rn}\boldsymbol{\Theta}^{\mathrm{FD}}\boldsymbol{h}_{\rm br}|^2}{\sigma_{\rm n}^2},\,\,
    \gamma_{\rm bf} \triangleq \frac{P_{\mathrm{BS}}|h_{\rm bf} + \boldsymbol{h}^H_{\rm rf}\boldsymbol{\Theta}^{\mathrm{FD}}\boldsymbol{h}_{\rm br}|^2}{\sigma_{\rm f}^2},\,\, \text{and}\,\, \gamma_{\rm d} \triangleq \frac{P_{\rm n} |h_{\rm nf} + \boldsymbol{h}^H_{\rm rf}\boldsymbol{\Theta}^{\mathrm{FD}}\boldsymbol{h}_{\rm nr}|^2}{\sigma_{\rm f}^2}, 
\end{equation}
Based on this, the feasibility conditions of problem $\mathrm{PC-FD}$ are presented in the following theorem.
\begin{theorem}
Problem $\mathrm{PC-FD}$ is feasible if and only if the following conditions hold.
\begin{subequations}
\label{eq:cdts_FD}
\begin{align}
    &\text{Condition 1:} \,\, \beta_{\min}^{\rm FD} \leq \beta_{\max}^{\rm FD}, \label{eq:cdt3}\\
    &\text{Condition 2:} \,\, \frac{\gamma_{\rm SI} t_{\rm n}^{\rm FD}}{ \gamma_{\rm bn}} \beta_{\min}^{\rm FD} + \frac{t_{\rm n}^{\rm FD}}{\gamma_{\rm bn}} \leq \frac{1}{2}, \label{eq:cdt4}
\end{align}
\end{subequations}
where $\beta_{\min}^{\rm FD}$ and $\beta_{\max}^{\rm FD}$ are expressed, respectively, as
\begin{equation}
    \label{FD_bounds11}
\beta_{\min}^{\rm HD} = \max \left(0, \frac{\gamma_{\rm bf} - t_{\rm f}^{\rm FD} - \frac{\gamma_{\rm bf}\left(1 + t_{\rm f}^{\rm FD} \right)t_{\rm n}^{\rm FD}}{\gamma_{\rm n}}}{c_1-c_2}  \right) \quad \text{and} \quad \beta_{\max}^{\rm HD} = \min \left(1, \frac{\gamma_{\rm n} - t_{\rm f}^{\rm FD} - t_{\rm n}^{\rm FD}\left(1 + t_{\rm f}^{\rm FD} \right)}{\gamma_{\rm SI}t_{\rm n}^{\rm FD}\left(1 + t_{\rm f}^{\rm FD} \right) + t_{\rm f}^{\rm FD}}\right),
\end{equation}
if $c_1 < c_2$ and 
\begin{equation}
    \label{FD_bounds12}
    \beta_{\min}^{\rm HD} = 0,\quad \text{and} \quad 
    \beta_{\max}^{\rm HD} = \min \left(1, \frac{\gamma_{\rm n} - t_{\rm f}^{\rm FD} - t_{\rm n}^{\rm FD}\left(1 + t_{\rm f}^{\rm FD} \right)}{\gamma_{\rm SI}t_{\rm n}^{\rm FD}\left(1 + t_{\rm f}^{\rm FD} \right) + t_{\rm f}^{\rm FD}},\frac{\gamma_{\rm bf} - t_{\rm f}^{\rm FD} - \frac{\gamma_{\rm bf}\left(1 + t_{\rm f}^{\rm FD} \right)t_{\rm n}^{\rm FD}}{\gamma_{\rm n}}}{c_1-c_2}\right),
\end{equation}
if $c_1 > c_2$,  such that $c_1 = \frac{\gamma_{\rm bf}}{\gamma_{\rm n}}\left(1 + t_{\rm f}^{\rm FD} \right)t_{\rm n}^{\rm FD} \gamma_{\rm SI}$, $c_2 =  \gamma_{\rm d}$, $t_{\rm n}^{\rm FD} = 2^{R_{\rm n}^{\rm th}} - 1$ and $t_{\rm f}^{\rm FD} = 2^{R_{\rm n}^{\rm th}} - 1$.
\end{theorem}
\begin{proof}
See Appendix C.
\end{proof}
Before continuing with the derivation of the optimal power control, let us define the following quantities. Let $\beta_c^{\rm FD}$, $\alpha_c^{\rm FD}$, $\alpha_{\min}^{\rm FD}$ and $\alpha_{\max}^{\rm FD}$ be defined, respectively, as 
\begin{equation}
\label{eq:vars}
\begin{aligned}
    &\beta_{\rm c}^{\rm FD} \triangleq \frac{t_{\rm f}^{\rm FD} \left(1 - \frac{\gamma_{\rm bn}}{\gamma_{\rm bf}} \right)}{\frac{\gamma_{\rm bn}}{\gamma_{\rm bf}}\gamma_{\rm SI} t_{\rm f}^{\rm FD} + \gamma_{d}}, \quad &\alpha_{\rm c}^{\rm FD} \triangleq \frac{ \gamma_{\rm SI} t_{\rm n}^{\rm FD}}{\gamma_{\rm bn}} \beta_{\rm c}^{\rm FD} + \frac{t_{\rm n}^{\rm FD}}{\gamma_{\rm bn}},\\
    &\alpha_{\min}^{\rm FD} \triangleq \frac{ \gamma_{\rm SI} t_{\rm n}^{\rm FD}}{ \gamma_{\rm bn}} \beta_{\min}^{\rm FD} + \frac{t_{\rm n}^{\rm FD}}{\gamma_{\rm bn}}, \quad &\alpha_{\max}^{\rm FD} \triangleq \frac{ \gamma_{\rm SI} t_{\rm n}^{\rm FD}}{ \gamma_{\rm bn}} \beta_{\max}^{\rm FD} + \frac{t_{\rm n}^{\rm FD}}{\gamma_{\rm bn}}, 
\end{aligned}
\end{equation}
and let $\alpha_{0}^{\rm FD}$ and $\beta_{0}^{\rm FD}$ be the quantities defined, respectively, as 
\begin{equation}
\label{eq:transformation}
\left(\alpha_{0}^{\rm FD}, \beta_{0}^{\rm FD}\right) = \left\{
\begin{aligned}
&\left(\alpha_{\min}^{\rm FD},\beta_{\min}^{\rm FD} \right), \quad &\text{if}\,\, \beta_c^{\rm FD} < \beta_{\min}^{\rm FD}, \\ 
&\left(\alpha_{\rm c}^{\rm FD},\beta_{\rm c}^{\rm FD} \right), \quad &\text{if}\,\, \beta_c^{\rm FD}  \in \left[\beta_{\min}^{\rm FD},\beta_{\max}^{\rm FD} \right],\\
&\left(\alpha_{\max}^{\rm FD},\beta_{\max}^{\rm FD} \right), \quad &\text{otherwise}.
\end{aligned}
\right.
\end{equation}
Based on the above definitions, let $\boldsymbol{\alpha}_{\rm n}^{\rm FD}$ and $\boldsymbol{\beta}^{\rm FD}$ be the vectors defined as 
\begin{equation}
\label{eq:intersection_points}
\begin{cases}
\begin{aligned}
\boldsymbol{\alpha}_{\rm n}^{\rm FD} \triangleq \left(\alpha_{\min}^{\rm FD},\alpha_{0}^{\rm FD},\alpha_{\max}^{\rm FD} \right) \quad \text{and} \quad \boldsymbol{\beta}^{\rm FD} \triangleq \left(\beta_{\min}^{\rm FD},\beta_{0}^{\rm FD},\beta_{\max}^{\rm FD} \right)
\end{aligned}, & \text{if} \quad \alpha_{0}^{\rm FD} \leq \frac{1}{2} \quad \text{and} \quad \alpha_{\max}^{\rm FD} \leq \frac{1}{2},\\
\begin{aligned}
\boldsymbol{\alpha}_{\rm n}^{\rm FD} \triangleq \left(\alpha_{\min}^{\rm FD},\alpha_{0}^{\rm FD} \right) \quad \text{and} \quad
\boldsymbol{\beta}^{\rm FD} \triangleq \left(\beta_{\min}^{\rm FD},\beta_{0}^{\rm FD} \right)
\end{aligned}, & \text{if} \quad \alpha_{0}^{\rm FD} \leq \frac{1}{2} \quad \text{and} \quad \alpha_{\max}^{\rm FD} \geq \frac{1}{2}, \\ 
\begin{aligned}
\boldsymbol{\alpha}_{\rm n}^{\rm FD} \triangleq \left(\alpha_{\min}^{\rm FD},\alpha_{\max}^{\rm FD} \right) \quad \text{and} \quad
\boldsymbol{\beta}^{\rm FD} \triangleq \left(\beta_{\min}^{\rm FD},\beta_{\max}^{\rm FD} \right)
\end{aligned}, & \text{if} \quad \alpha_{0}^{\rm FD} \geq \frac{1}{2} \quad \text{and} \quad \alpha_{\max}^{\rm FD} \leq \frac{1}{2},\\
\begin{aligned}
\boldsymbol{\alpha}_{\rm n}^{\rm FD} \triangleq \alpha_{\min}^{\rm FD}  \quad \text{and} \quad 
\boldsymbol{\beta}^{\rm FD} \triangleq \beta_{\min}^{\rm FD} 
\end{aligned}, & \text{if} \quad \alpha_{0}^{\rm FD} \geq \frac{1}{2} \quad \text{and} \quad \alpha_{\max}^{\rm FD} \geq \frac{1}{2}.
\end{cases}
\end{equation}
Let $L \in \{1,2,3\}$ be the number of elements of  $\boldsymbol{\alpha}_{\rm n}^{\rm FD}$ (or $\boldsymbol{\beta}^{\rm FD}$ equivalently). Then, let $\mathbf{x}_1^{\rm FD}$ and $\mathbf{x}_2^{\rm FD}$ be the $2 \times L$ vectors defined, respectively, as $\mathbf{x}_1^{\rm FD} \triangleq  t_{\rm f}^{\rm FD}\boldsymbol{\alpha}_{\rm n}^{\rm FD} + \frac{ \gamma_{\rm SI} t_{\rm f}^{\rm FD}}{\gamma_{\rm bn}} \boldsymbol{\beta}^{\rm FD}   + \frac{t_{\rm f}^{\rm FD}}{ \gamma_{\rm bn}}$ and $\mathbf{x}_2^{\rm FD} \triangleq t_{\rm f}^{\rm FD}\boldsymbol{\alpha}_{\rm n}^{\rm FD} - \frac{ \gamma_{\rm d}}{ \gamma_{\rm bf}} \boldsymbol{\beta}^{\rm FD}   + \frac{t_{\rm f}^{\rm FD}}{ \gamma_{\rm bf}}$. Finally, let $\boldsymbol{\alpha}_{\rm f}^{\rm FD}$ be the $1 \times L$ vector defined as
\begin{equation}
    \label{eq:power_far}
   [\boldsymbol{\alpha}_{\rm f}^{\rm FD}]_{i}   = \max \left([\mathbf{x}_1^{\rm FD}]_{i},[\mathbf{x}_2^{\rm FD}]_{i} \right), \quad \forall\,\,i \in \llbracket 1, L \rrbracket
\end{equation}
To this end, based on the above and assuming that problem $\mathrm{PC-FD}$ is feasible, its optimal solution is given in the following theorem. 
\begin{theorem}
Assuming that problem $\mathrm{PC-FD}$ is feasible, i.e., conditions \eqref{eq:cdt3} and \eqref{eq:cdt4} hold, its optimal solution is given by
\begin{equation}
    \left(\alpha_{\rm n}^{{\rm FD}^*}, \alpha_{\rm f}^{{\rm FD}^*},\beta^{{\rm FD}^*} \right) = \argmin_{\mathbf{p}^{\rm FD} \in \mathcal{P}^{\rm FD}}  f \left(\mathbf{p}^{\rm FD} \right),
\end{equation}
where $\mathcal{P}^{\rm FD} = \left\{\mathbf{p}_1^{\rm FD},\dots,\mathbf{p}_L^{\rm FD} \right\}$, where for $i \in \left[1,L \right]$, $\mathbf{p}_i^{\rm FD}$ is the $1 \times 3$ vector expressed as $\mathbf{p}_i^{\rm FD} = \left([\boldsymbol{\alpha}_{\rm n}^{\rm FD}]_{i},[\boldsymbol{\alpha}_{\rm f}^{\rm FD}]_{i},[\boldsymbol{\beta}^{\rm FD}]_{i} \right)$.
\end{theorem}
\begin{proof}
See Appendix D.
\end{proof}
\subsection{RIS-enabled FD C-NOMA: Phase-Shift Coefficients Optimization}
Let $\boldsymbol {v} = [v_{1}, \dots, v_{M}]^{H}$, where for all $m \in \llbracket 1, M \rrbracket$, $v_{m} = e^{j\theta _{m}}$. By applying the change of variables $\boldsymbol {h}^{H}_{\rm rn} \boldsymbol{\Theta}^{\mathrm{FD}} \boldsymbol {h}_{\rm br} = \boldsymbol {v}^{H}\boldsymbol {\Phi }$, where $\boldsymbol {\Phi} = \text {diag} (\boldsymbol {h}^{H}_{\rm rn})\boldsymbol {h}_{\rm br} \in \mathbb {C}^{M \times 1}$, $\boldsymbol {h}^{H}_{\rm rf} \boldsymbol{\Theta}^{\mathrm{FD}}\boldsymbol {h}_{\rm br} = \boldsymbol{v}^{H} \boldsymbol{\Psi}$, where $\boldsymbol {\Psi} = \text {diag}(\boldsymbol{h}^{H}_{\rm rf})\boldsymbol {h}_{\rm br} \in \mathbb {C}^{M \times 1}$, and $\boldsymbol {h}^{H}_{\rm rf} \boldsymbol{\Theta}^{\mathrm{FD}} \boldsymbol {h}_{\rm nr} = \boldsymbol {v}^{H}\boldsymbol {\Xi}$, where $\boldsymbol{\Xi } = \text {diag}(\boldsymbol {h}^{H}_{\rm rf})\boldsymbol {h}_{\rm nr} \in \mathbb {C}^{M \times 1}$, we have 
\begin{align}
|h_{\rm bn} + \boldsymbol {h}^{H}_{\rm rn} \boldsymbol{\Theta}^{\mathrm{FD}} \boldsymbol {h}_{\rm br}|^{2} &= |h_{\rm bn} + \boldsymbol {v}^{H}\boldsymbol{\Phi}|^{2}, \cr |h_{\rm bf} + \boldsymbol {h}^{H}_{\rm rf} \boldsymbol{\Theta}^{\mathrm{FD}} \boldsymbol {h}_{\rm br}|^{2} &= |h_{\rm bf} + \boldsymbol {v}^{H}\boldsymbol {\Psi}|^{2}, \cr 
|h_{\rm nf} + \boldsymbol {h}^{H}_{\rm rf} \boldsymbol{\Theta}^{\mathrm{FD}} \boldsymbol {h}_{\rm nr}|^{2} &=|h_{\rm nf} + \boldsymbol {v}^{H}\boldsymbol {\Xi}|^{2}.
\end{align}
Then, the achievable data rates in \eqref{Eq: FD R_n_f}, \eqref{Eq: FD R_n_n}, and \eqref{Eq: FD R_MRC} can be expressed as follows
\begin{align}
    \mathcal{R}_{\rm n \longrightarrow f}^{\mathrm{FD}} = &~\mathrm{log}\left( 1 + \frac{\alpha_{\rm f}^{\mathrm{FD}} P_{\mathrm{BS}} (\bar{\boldsymbol {v}}^{H}\boldsymbol {Q}_{\rm bn}\bar{\boldsymbol {v}} + | {h}_{\rm bn}|^{2})}{\alpha_{\rm n}^{\mathrm{FD}} P_{\mathrm{BS}} (\bar{\boldsymbol {v}}^{H}\boldsymbol {Q}_{\rm bn}\bar{\boldsymbol {v}} + | {h}_{\rm bn}|^{2}) + \beta^{\mathrm{FD}} P_{n} \gamma_{\rm{SI}} + \sigma^2_{\rm n}}\right),\\
    \mathcal{R}_{n \longrightarrow n}^{\mathrm{FD}} = &~\mathrm{log}\left( 1 + \frac{\alpha_{\rm n}^{\mathrm{FD}} P_{\mathrm{BS}} (\bar{\boldsymbol {v}}^{H}\boldsymbol {Q}_{\rm bn}\bar{\boldsymbol{v}} + |{h}_{\rm bn}|^{2})}{\beta^{\mathrm{FD}} P_{n} \gamma_{\rm{SI}} + \sigma_{\rm n}^2}\right), \\
    \mathcal{R}_{\mathrm{MRC}}^{\mathrm{FD}} = & \mathrm{log}\left(1 + \frac{\alpha_{\rm f}^{\mathrm{FD}} P_{\mathrm{BS}} (\bar{\boldsymbol {v}}^{H}\boldsymbol {Q}_{\rm bf}\bar{\boldsymbol {v}} + |{h}_{\rm bf}|^{2}) + \beta^{\mathrm{FD}}P_{\rm n} (\bar{\boldsymbol {v}}^{H}\boldsymbol {Q}_{\rm nf}\bar{\boldsymbol {v}} + | {h}_{\rm nf}|^{2})}{\alpha_{\rm n}^{\mathrm{FD}} P_{\mathrm{BS}} (\bar{\boldsymbol {v}}^{H}\boldsymbol {Q}_{\rm bf}\bar{\boldsymbol {v}} + | {h}_{\rm bf}|^{2}) + 1}\right),
\end{align}
where 
\begin{equation} \boldsymbol {Q}_{\rm bn}=\begin{bmatrix} \boldsymbol {\Phi }\boldsymbol {\Phi }^{H} &\boldsymbol {\Phi } {h}^H_{\rm bn} \\  {h}_{\rm bn}\boldsymbol {\Phi }^{H} &0 \\ \end{bmatrix}, ~ \boldsymbol {Q}_{\rm bf}=\begin{bmatrix} \boldsymbol {\Psi }\boldsymbol {\Psi }^{H} &\boldsymbol {\Psi } {h}^H_{\rm bf} \\ {h}_{\rm bf}\boldsymbol {\Psi }^{H} & 0 \\ \end{bmatrix}, ~ \boldsymbol {Q}_{\rm nf}=\begin{bmatrix} \boldsymbol {\Xi }\boldsymbol {\Xi }^{H} &\boldsymbol {\Xi } {h}^H_{\rm nf} \\ {h}_{\rm nf}\boldsymbol {\Xi }^{H} &0 \\ \end{bmatrix}, ~ \bar{\boldsymbol {v}}=\begin{bmatrix} \boldsymbol {v} \\ t \\ \end{bmatrix}.\end{equation} 
\par Note that $\bar{\boldsymbol {v}}^{H}\boldsymbol {Q}_{z}\bar{\boldsymbol {v}}={\mathrm{tr}}(\boldsymbol {Q}_{z}\bar{\boldsymbol {v}}\bar{\boldsymbol {v}}^{H})$ for all $k \in \{{\rm bn}, {\rm bf}, \rm nf\}$ and define $\boldsymbol {V}=\bar{\boldsymbol {v}}\bar{\boldsymbol {v}}^{H}$, which needs to satisfy ${\mathrm{rank}}(\boldsymbol {V})=1$ and $\boldsymbol {V}\succeq \boldsymbol {0}$. However, the rank-one constraint is non-convex \cite{Zhang_TWC_2019}. By relaxing this constraint, the optimization problem $\mathrm{PS-FD}$ can be transformed into 
\allowdisplaybreaks
\begingroup
\begin{subequations}
\label{prob:SDR}
\begin{align}
&\hat{\mathcal{P}}: \mathrm{Find}~\boldsymbol{\theta}^{\mathrm{FD}}, \\
&\text{s.t.}\,\,\alpha_{\rm n}^{\mathrm{FD}}P_{\mathrm{BS}} ({\mathrm{tr}}(\boldsymbol {Q}_{\rm bn}\boldsymbol{V}) + |h_{\rm bn}|^2) \geq t^{\mathrm{FD}}_{\rm n} (\beta^{\mathrm{FD}} P_{\rm n}\gamma_{\mathrm{SI}} + \sigma_{\rm n}^2), \label{P2_C1}\\
&\quad \,\,\,  {\mathrm{tr}}(\boldsymbol {Q}\boldsymbol{V}) + \alpha_{\rm f}^{\mathrm{FD}}P_{{\mathrm{BS}}} |h_{\rm bf}|^2 + \beta^{\mathrm{FD}} P_{\rm n}|h_{\rm nf}|^2 \geq t^{\mathrm{FD}}_{\rm f}(\alpha_{\rm n}^{\mathrm{FD}}P_{\mathrm{BS}}({\mathrm{tr}}(\boldsymbol{Q}_{\rm bf}\boldsymbol{V}) + |h_{\rm bf}|^2) + \sigma_{\rm f}^2), \label{P2_C2}\\ 
&\quad \,\,\,\alpha_{\rm f}^{\mathrm{FD}}P_{\mathrm{BS}} ({\mathrm{tr}}(\boldsymbol {Q}_{\rm bn}\boldsymbol{V}) + |h_{\rm bn}|^2) \geq t^{\mathrm{FD}}_{\rm f} (\alpha_{\rm n}^{\mathrm{FD}}P_{\mathrm{BS}}({\mathrm{tr}}(\boldsymbol {Q}_{\rm bn}\boldsymbol{V}) + |h_{\rm bn}|^2) + \beta^{\mathrm{FD}} P_{\rm n}\gamma_{\mathrm{SI}}+ \sigma_{\rm f}^2), \label{P2_C4} \\&\quad\,\,\, \boldsymbol {V} \succeq 0, \\ &\quad\,\,\, [\boldsymbol{V}]_{m,m} = 1, \qquad\qquad \forall\,\, m \in \llbracket 1, M+1 \rrbracket, 
\end{align}
\end{subequations}
\endgroup
where $\boldsymbol{Q} = \alpha_{\rm f}^{\mathrm{FD}}P_{\mathrm{BS}} \boldsymbol{Q}_{\rm bf} + \beta^{\mathrm{FD}} P_{\rm n} \boldsymbol{Q}_{\rm nf}$. It is not difficult to observe that problem $\hat{\mathcal{P}}$ is an SDP. Thus, an optimal solution can be obtained by existing convex optimization solvers such as CVX \cite{Zhang_TWC_2019}. However, due to the rank one constraint relaxation, the SDR may not be tight for $\mathrm{PS-FD}$. To overcome this issue, the Gaussian randomization approach can be similarly used to find a feasible solution to problem $\mathrm{PS-FD}$ dependent on the higher-rank solution that is obtained by solving problem $\hat{\mathcal{P}}$. Following the same procedures in \textbf{Algorithm 1}, the solutions for the phase-shift coefficients, $\boldsymbol{\theta}^{\mathrm{FD}}$, and the power control coefficients, $(\alpha_{\rm n}^{\mathrm{FD}}, \alpha_{\rm f}^{\mathrm{FD}}, \beta^{\mathrm{FD}})$ can be obtained.
\section{Results and Discussion}
\label{Sec:Simulation}
In this section, several numerical examples and simulation results are presented to examine the performance of the proposed schemes RIS-enabled FD C-NOMA and RIS-enabled HD C-NOMA networks. Since it has been proven that the performance of traditional FD C-NOMA has better performance than the HD C-NOMA counterpart, we consider the traditional FD C-NOMA technique without any assistance of RIS as a benchmark \cite{Dinh_2020,Zhang_2017_Full, Zhong_2016_Non}. This technique is referred to as "FD C-NOMA without RIS". In this scheme, the BS communicates with UE$_{\rm n}$ and UE$_{\rm f}$ and UE$_{\rm n}$ communication with UE$_{\rm f}$ in the same time-slot without any aid from the RIS. Hence, with the same target objective (power consumption minimization), we only need to use the power allocation coefficients at the BS and UE$_{\rm n}$ to minimize the total transmit power, which are obtained by \textit{\textbf{Theorem 4}} considering only the direct channel gains (BS $\longrightarrow$ UE$_{\rm n}$, BS $\longrightarrow$ UE$_{\rm f}$, and UE$_{\rm n}$ $\longrightarrow$ UE$_{\rm f}$) in the derived closed-form expressions.
\subsection{Simulation Settings}
The simulation environment consists of one BS, one RIS, one UE$_{\rm n}$, and one UE$_{\rm f}$ which are located in a 3-dimensional Cartesian coordinates systems $\left(X,Y,Z \right)$ at $(0\,\text{m}, 10\,\text{m}, 0\,\text{m}),$ $(80\,\text{m}, 10\,\text{m}, 0\,\text{m}),$ $(40\,\text{m}, 0\,\text{m}, 0\,\text{m})$ and $(80\,\text{m}, 0\,\text{m}, 0\,\text{m})$, respectively. Note that the RIS is located nearby UE$_{\rm f}$ to enhance its communication link and correspondingly its achievable data rate, which reflects on minimizing the required transmission power. Moreover, the SI at UE$_{\rm n}$ is modeled as a Rayleigh fading with zero mean and $\Omega_{\mathrm{SI}}$ variance \cite{Liu_2017_Hybrid}. Note that, Monte-Carlo simulations are employed over $10^4$ independent channel realizations. The system parameters for the simulations are listed in Table \ref{T2} \cite{zuo2020reconfigurable, Guo_2020_Intelligent, Elhattab_2019_A}.
\vspace{-0.8cm}
\begin{table}[t]
\caption{Simulation Parameters}
\centering
\renewcommand{\arraystretch}{0.5} 
\setlength{\tabcolsep}{0.2cm}  
\begin{tabular}{| c | c | c |}
  \hline 
  Parameter & Symbol & Value \\
  \hline
  Rician factors of the BS-RIS and the RIS-UE$_{\rm f}$ links & $\kappa_{\rm br}, \kappa_{\rm rf}$  & 3 $\mathrm{dB}$ \\
  \hline
  Path-loss at reference distance of 1 m & $\rho_0$ & $-30$ dB\\ 
  \hline
  Path-loss exponents for the BS-RIS and for the RIS-UE$_{\rm f}$ links & $\eta_{\rm br}, \eta_{\rm rf}$ & 2.2 \\ 
  \hline
  Path-loss exponents for the BS-UE$_{\rm f}$ and for the UE$_{\rm n}$-UE$_{\rm f}$ links& $\eta_{\rm bf}, \eta_{\rm nf}$ & 4 \\
  \hline 
  Path-loss exponent for the  BS-UE$_{\rm n}$ link & $\eta_{\rm bn}$ & 3.5 \\
  \hline 
  Path-loss exponent for the UE$_{\rm n}$-RIS link & $\eta_{\rm nr}$ & 3 \\ 
  \hline
  Power budget at the BS & $P_{\mathrm{BS}}$ & 33 dBm \\
  \hline 
  Power budget at UE$_{\rm n}$ & $P_{\rm n}$ & 23 dBm \\
  \hline
  Noise power at UE$_{\mathrm{n}}$ and UE$_{\mathrm{f}}$ & $\sigma_{\rm n}^2, \sigma_{\rm f}^2$ & $- 90$ dB \\ 
  \hline 
  Minimum rate QoS requirement for UE$_{\rm n}$ & $R_{\mathrm{n}}^{\mathrm{th}}$ & 1 bits/sec/Hz \\ 
  \hline
\end{tabular} 
\label{T2}
\end{table}
\subsection{Convergence of The Proposed Schemes}
\begin{figure}[!tbp]
  \centering
  \begin{minipage}[b]{0.45\textwidth}
  \centering
    \includegraphics[width=\textwidth]{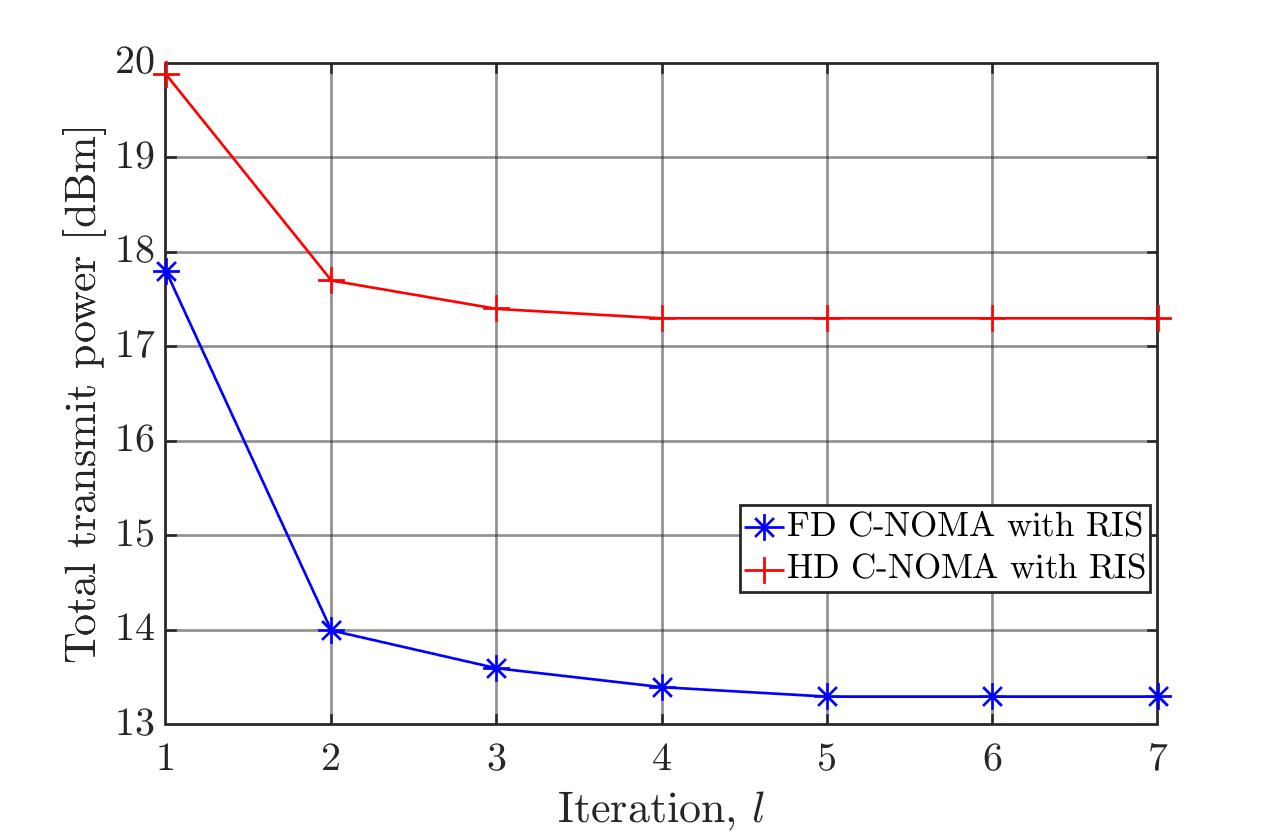}
    \caption{Convergence of the proposed algorithm.}\label{fig:Convergence}
  \end{minipage}
  \hfill
  \begin{minipage}[b]{0.45\textwidth}
    \includegraphics[width=\textwidth]{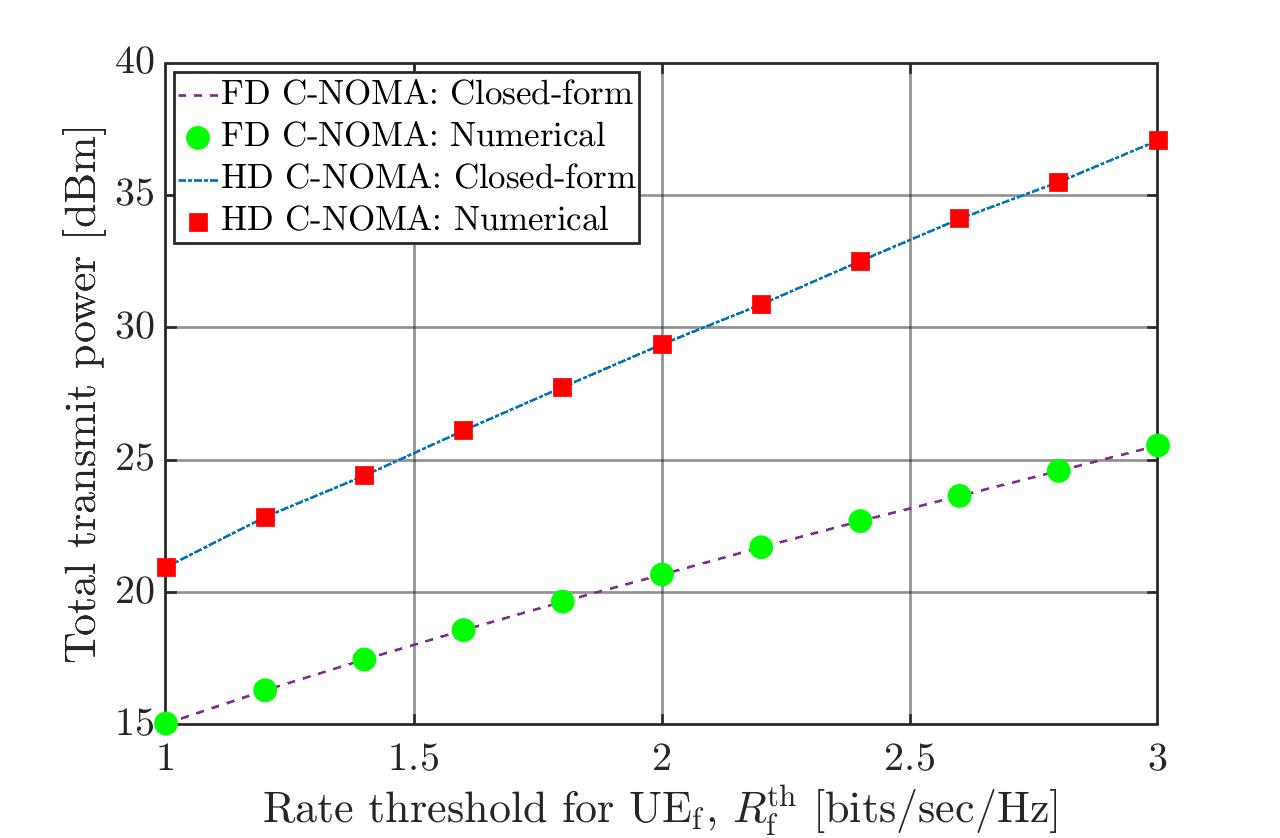}
    \caption{Analytical  and  numerical  total transmit power  versus  the rate threshold for UE$_{\rm f}$.} \label{fig:validation}
  \end{minipage}
\end{figure}
\par Fig. \ref{fig:Convergence} describes the convergence behavior of the proposed RIS-enabled FD C-NOMA and RIS-enabled HD C-NOMA algorithms versus the iteration number with number of RIS elements $M = 80$, SI parameter $\Omega_{\rm SI}= -100$ dB, and minimum required rate for UE$_{\rm f}$  $R_{\rm f}^{\rm th}= 2$ bits/sec/Hz. It can be observed that the proposed RIS-enabled FD C-NOMA and RIS-enabled HD C-NOMA algorithms converge in about $4$ to $5$ iterations, which provides a low computational complexity \cite{Zhang_TWC_2019}.
\subsection{Validation of The Closed-form Expressions for The Power Allocation Coefficients}
It can be seen that the closed-form expressions in \textbf{\textit{Theorem 2}} and \textbf{\textit{Theorem 4}} are derived for a given phase-shift matrix. As a result, in order to validate the closed-form expression for the power allocation coefficients, we consider FD C-NOMA without RIS and HD- C-NOMA without RIS as the schemes that validate the analytical results. Fig. \ref{fig:validation} depicts the analytical and numerical total transmit power for FD C-NOMA without RIS and HD C-NOMA without RIS schemes versus the minimum required rate for UE$_{\rm f}$, $R_{\rm f}^{\mathrm{th}}$ with $\Omega_{\mathrm{SI}} = - 90$ dB. The analytical results are obtained based on the closed-from expressions derived in \textit{\textbf{Theorem 2}} and \textit{\textbf{Theorem 4}}, while the numerical results are obtained by solving problem $\mathrm{PC-HD}$ and $\mathrm{PC-FD}$ using an off-the-shelf optimization solver. \footnote{The adopted solver is fmincon that is a predefined \MATLAB solver \cite{Dinh_2020,Elhattab_2020_Power, ebbesen2012generic}. Moreover, $10^3$ different initial points are generated to guarantee the convergence of the solver to the optimal solution.} It can be seen from Fig. \ref{fig:validation} that the analytical results match perfectly the  numerical results which validates the optimality of the closed-form expressions of the power allocation coefficients obtained by \textbf{\textit{Theorem 2}} and \textbf{\textit{Theorem 4}}.
\begin{figure}[!t]
    \centering
     \includegraphics[width = 0.45 \textwidth]{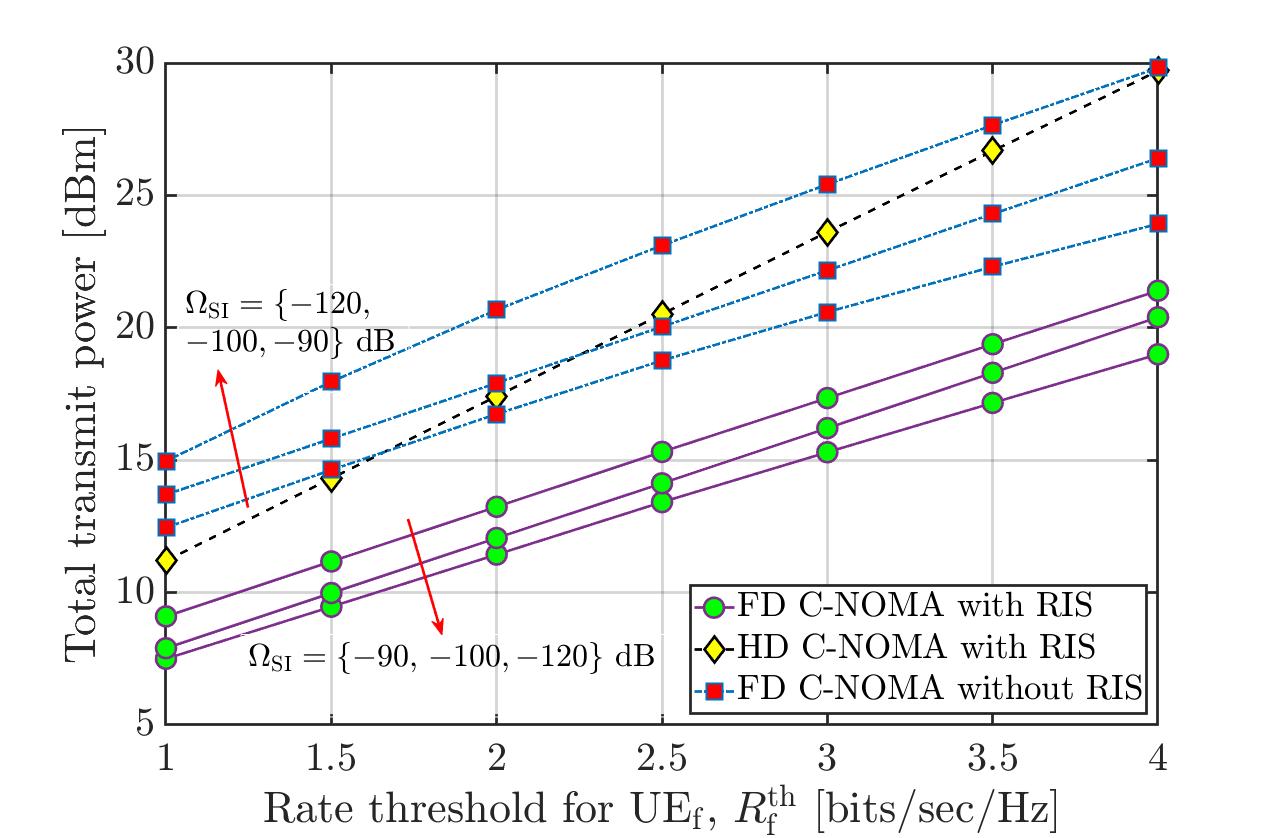}
    \caption{Total transmit power versus the rate threshold $R_{\rm f}^{\rm{th}}$ of UE$_{\rm f}$ with different SI values due to the FD operation $(M = 80)$.}
    \label{fig:Rth}
\end{figure}
\subsection{Total Transmit Power Performance}
Fig. \ref{fig:Rth} compares the performance of the proposed RIS-enabled FD C-NOMA scheme, the proposed RIS-enabled HD C-NOMA, and the FD C-NOMA without RIS adopted in \cite{Zhang2019Resource,Zhang_2017_Full, Zhong_2016_Non}, when varying the required data rate threshold for UE$_{\rm f}$. Different values of the SI channel gain parameter $\Omega_{\mathrm{SI}}$ are considered for the cases of RIS-enabled FD C-NOMA and FD C-NOMA without RIS schemes. First, it can be seen that FD C-NOMA with RIS scheme has a significant gain compared to HD C-NOMA with RIS and FD C-NOMA without RIS. Second, due to the pre-log penalty in the HD scenario, the gap between the FD C-NOMA with RIS and the HD C-NOMA with RIS increases when the required data rate threshold increases. Third, since the RIS constructs strong combined channel gain for the BS-UE$_{\rm f}$ and UE$_{\rm n}$-UE$_{\rm f}$ communication links, this makes the active nodes (BS and UE$_{\rm n}$) able to transmit with low power, and hence, the effects of the SI on the system performance is weak compared to the FD C-NOMA without RIS. Finally, it can be seen that the HD C-NOMA with RIS has the ability to beat the FD C-NOMA without RIS proposed in \cite{Zhang2019Resource,Zhang_2017_Full, Zhong_2016_Non}. Depending on the strength of the SI, the HD C-NOMA with RIS can achieve a better performance for most of the considered cases (different rate thresholds for UE$_{\mathrm{f}}$).
\begin{figure}[!t]
\centering
\subfigure[]{
    \includegraphics[width=0.45\columnwidth ] {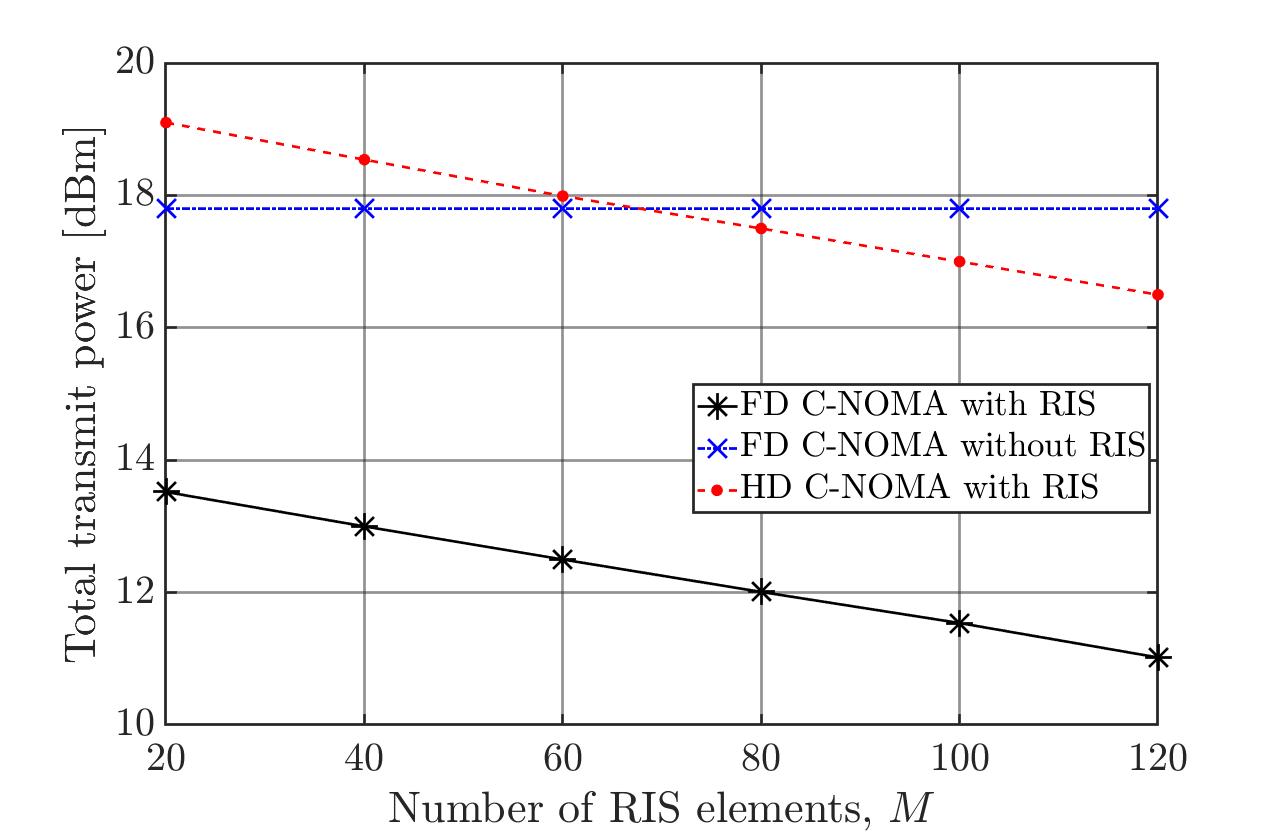}
}
\subfigure[]{
   \includegraphics[width=0.45\columnwidth ]{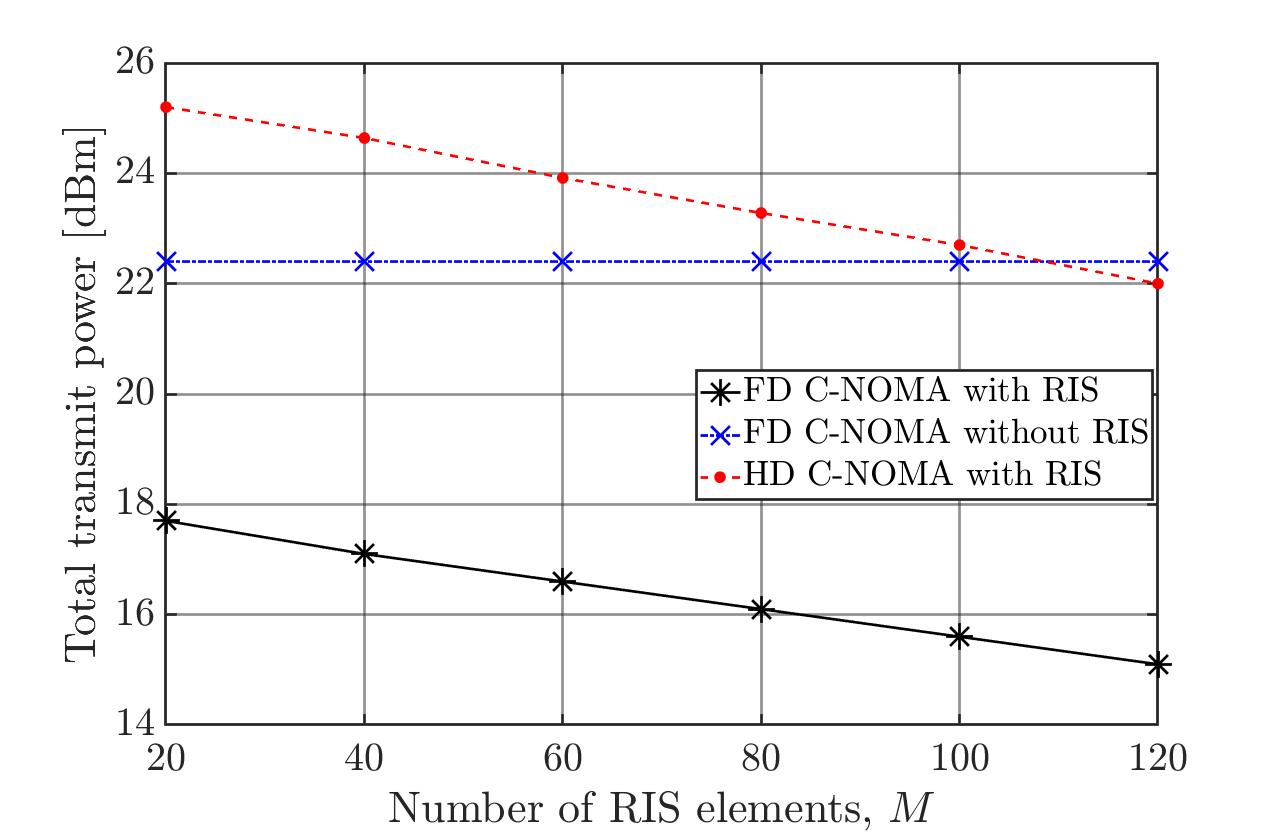}
}
\caption[]{Total transmit power versus the number of RIS elements, when $\Omega_{\mathrm{SI}} = -100$ dB, where (a) for the case when  $R^{\mathrm{th}}_{\rm f} = 2$ bits/s/Hz and (b) for the case when  $R^{\mathrm{th}}_{\rm f} = 3$ bits/s/Hz.}
\label{Fig: Power_vs_RIS elements}
\end{figure}
\par Fig. \ref{Fig: Power_vs_RIS elements} presents the total transmit power for the three schemes versus the number of RIS reflecting elements, and $\Omega_{\mathrm{SI}} = -100$ dB, where the rate QoS constraint at UE$_{\rm f}$ $R^{\mathrm{th}}_{\rm f} = 2$ bits/s/Hz in Fig. \ref{Fig: Power_vs_RIS elements}.a and $R^{\mathrm{th}}_{\rm f} = 3$ bits/s/Hz in Fig. \ref{Fig: Power_vs_RIS elements}.b. First, it is observed that the total transmit power that is required by the RIS-based schemes decreases when the number of meta-atoms of the RIS increases while the total transmit power for FD C-NOMA without RIS scheme  remains unchanged. This is because a larger number of RIS meta-atoms leads to a higher combined channel gains and hence a higher passive array gains. Second, it can be seen that the number of meta-atoms required at the RIS to allow the HD C-NOMA with RIS to beat the traditional FD C-NOMA depends on the required QoS at UE$_{\rm f}$.  This is because a high QoS requirement at UE$_{\rm f}$ needs a high passive array gain to overcome the pre-log penalty in the HD mode. Third, the RIS-enabled FD C-NOMA network significantly outperforms the other schemes, which reveals the potential of integrating RIS in FD C-NOMA networks in enhancing the network power efficiency.
\begin{figure}[!t]
\centering
\subfigure[ ]{
    \includegraphics[width=0.46\columnwidth ]{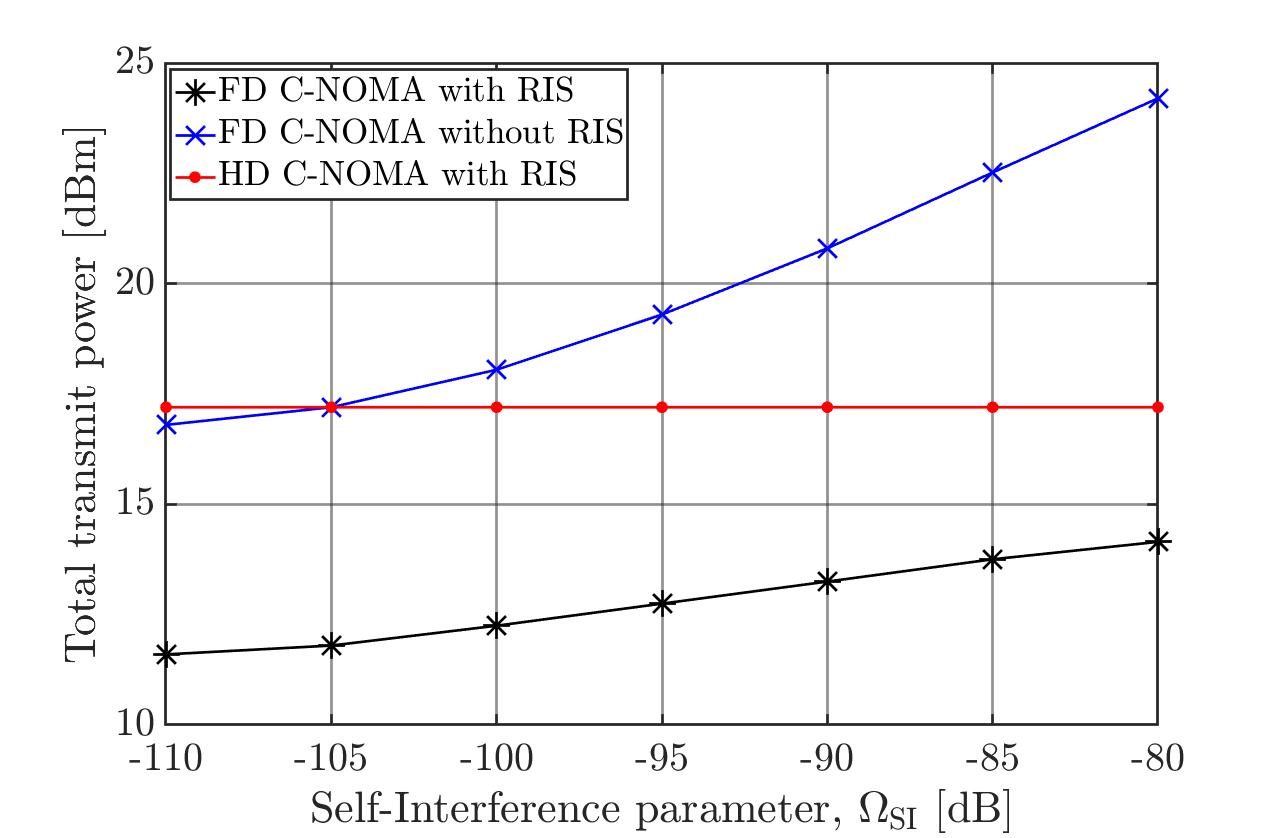}
    }
\subfigure[]{
   \includegraphics[width=0.46\columnwidth ]{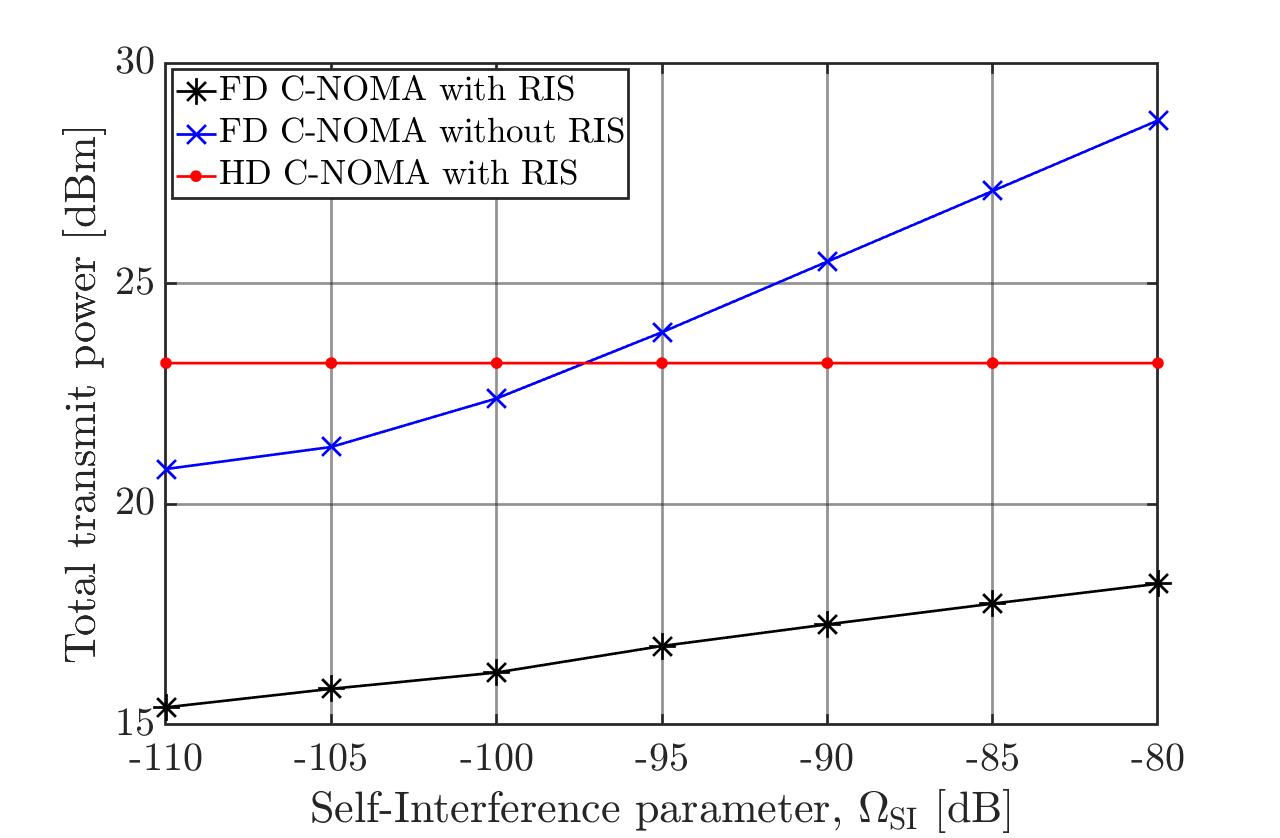}
}
\caption[]{Total transmit power versus the SI parameter at UE$_{\rm n}$, when the number of reflecting elements at the RIS is $M = 80$ dB, where (a) for the case when  $R^{\mathrm{th}}_{\rm f} = 2$ bits/s/Hz and (b) for the case when  $R^{\mathrm{th}}_{\rm f} = 3$ bits/s/Hz.}
\label{fig: power vs SI}
\end{figure}
\par Fig. \ref{fig: power vs SI}.a and Fig. \ref{fig: power vs SI}.b depict the total transmit power versus the SI values at the near NOMA user for $R_{\mathrm{f}}^{\mathrm{th}} = 2$ and $R_{\mathrm{f}}^{\mathrm{th}} = 3$, respectively, when $M = 80$. First, it can be observed that the proposed FD C-NOMA with RIS scheme gives a significant performance enhancement than the HD C-NOMA with RIS scheme when $\Omega_{\mathrm{SI}}$ is relatively small. However, as $\Omega_{\mathrm{SI}}$ increases, the performance gain between them decreases. This is because increasing the SI value restricts the power transmission at UE$_{\rm n}$ and hence the BS should increase its transmit power in order to meet the QoS constraint at UE$_{\rm f}$. Second, the HD C-NOMA with RIS scheme can achieve a significant performance compared to FD C-NOMA without RIS when $\Omega_{\mathrm{SI}}$ increases. This is because increasing $\Omega_{\mathrm{SI}}$ leads to deteriorating the performance of the FD mode and hence the passive array gain at the RIS can make the HD C-NOMA with RIS be a favorable scheme compared to the FD C-NOMA without RIS scheme. Third, it can be also observed that the total transmit power in the network without RIS is sharply increasing in comparison with the rate of increase in the network with RIS. This means that the system with RIS can tolerate high values of SI, which validates the effectiveness of the amalgamation between FD C-NOMA and RIS.  
\par Fig. \ref{fig:power transmit} shows the effect of increasing the SI on the transmit power at both the BS and the near NOMA user for different numbers of RIS elements. It can be seen that the RIS elements have a great impact on the amount of transmit power at the active nodes (BS and UE$_{\rm n}$). In addition, it can be also seen that when $\Omega_{\mathrm{SI}}$ increases, UE$_{\rm n}$ should reduce its transmit power to avoid degrading its performance. On the other hand, reducing the power at UE$_{\rm n}$ will lead to a sharp increase at the BS side, which shed the lights on the effect of the D2D communication on the network performance in terms of the power efficiency. For instance, for $M = 80$, when $\Omega_{\mathrm{SI}}$ increases from $-100$ dB to $-90$ dB, the transmit power at UE$_{\rm n}$ reduces by $0.9$ dB, whereas the transmit power at the BS increases by $3$ dB.
\begin{figure}[!t]
    \centering
     \includegraphics[width = 0.45 \textwidth]{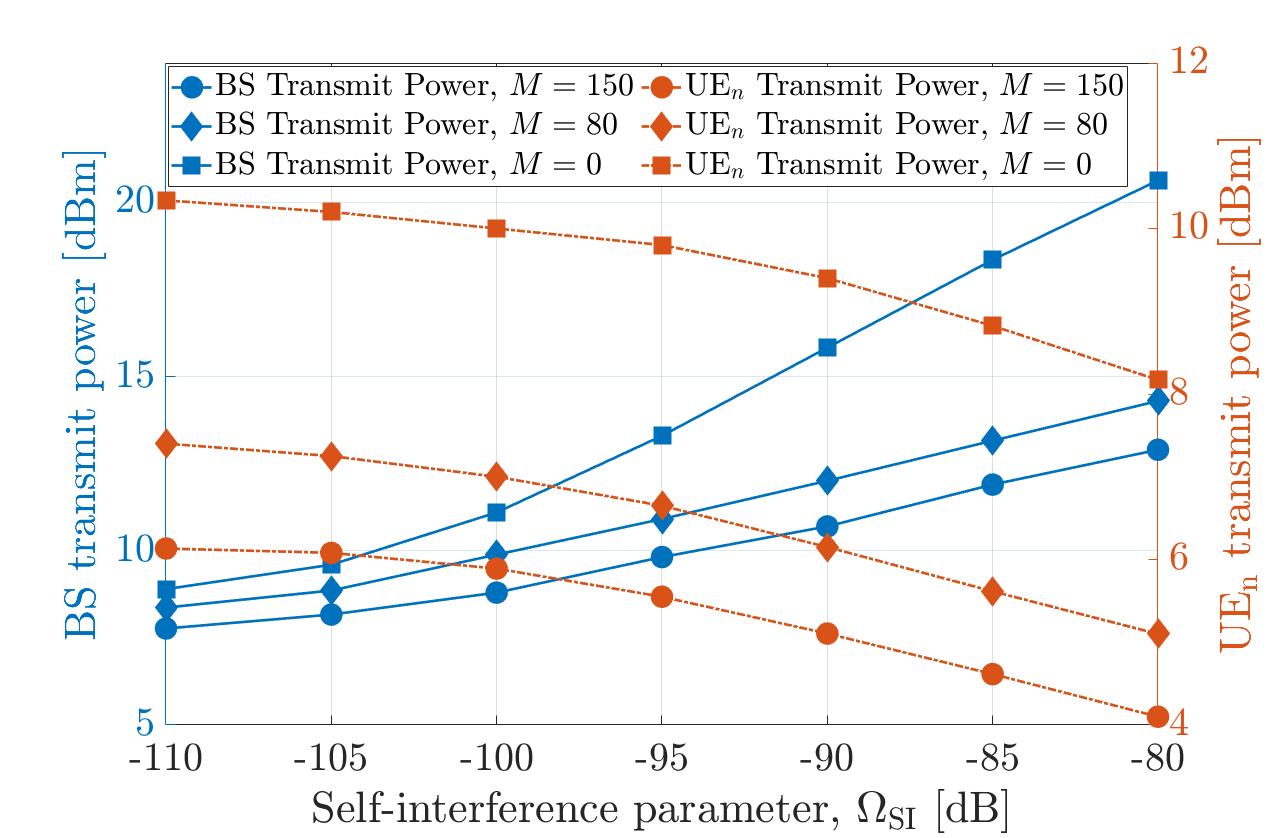}
    \caption{Total transmit power by the BS and UE$_{\rm n}$ versus the SI channel gains with different number of RIS elements for the proposed FD C-NOMA with RIS.}
    \label{fig:power transmit}
\end{figure}
\section{Conclusion}
\label{Sec:Conclusion}
In this paper, we investigated the RIS-empowered HD/FD C-NOMA downlink transmission scheme in order to minimize the total transmit power at the BS and at the relay user by jointly optimizing the power allocation coefficients at the BS and the transmit relaying power at the near NOMA user, along with the phase-shift coefficients at the RIS. In this formulated problem, we considered the QoS requirements in terms of the minimum data rate, SIC constraint, the power budget at the BS and the near NOMA user, and the RIS elements constraints. By invoking the alternating optimization technique, the non-convex power minimization optimization problem is decomposed into two sub-problems, power allocation optimization sub-problem and phase-shift optimization sub-problem, which are solved in an alternate manner. By leveraging the SDR, the RIS phase-shift coefficients are obtained. Meanwhile, for the power allocation sub-problem, we derived closed-form expressions for the optimal power allocation coefficients at the BS and the optimal transmit-relaying power from the near NOMA user.  The Simulation results show that the proposed RIS-enabled FD C-NOMA scheme significantly outperforms the FD C-NOMA without the assistance of the RIS. In addition, in spit of the pre-log penalty in the HD relaying mode, and according to the required QoS at the far NOMA user, the number of reflecting elements and the SI value, the proposed RIS-enabled HD C-NOMA can outperform the FD C-NOMA without RIS.  
\appendices 
\section{Proof of Theorem 1}
In this section, we present the proof of \textbf{\textit{Theorem 1}}. At the beginning, constraints \eqref{P1_C1}-\eqref{P1_C6} can be expressed as
\begin{subequations}
\begin{align}
&0 \leq \alpha_{\rm n}^{\mathrm{HD}} \leq \alpha_{\rm f}^{\mathrm{HD}} \leq 1,\label{P1_C11}\\
&\alpha_{\rm f}^{\mathrm{HD}}  \leq 1-\alpha_{\rm n}^{\mathrm{HD}},\label{P1_C12} \\
&0 \leq \beta^{\mathrm{HD}} \leq 1,\label{P1_C13} \\
&\alpha_{\rm n}^{\rm HD} \geq \frac{t_{\rm n}^{\rm HD}}{\gamma_{\rm bn}},\label{P1_C14} \\ 
&\alpha_{\rm f}^{\rm HD} \geq \alpha_{\rm n}^{\rm HD} t_{\rm f}^{\rm HD} + \frac{t_{\rm f}^{\rm HD}}{\gamma_{\rm bn}},\label{P1_C15} \\ 
&\beta^{\rm HD} \geq \frac{t_{\rm f}^{\rm HD}}{\gamma_{\rm d}} - \frac{\alpha_{\rm f}^{\rm HD} \gamma_{\rm bf}}{\gamma_{\rm d}\left(\alpha_{\rm n}^{\rm HD}\gamma_{\rm bf}+1 \right)}\label{P1_C16}.
\end{align}
\end{subequations}
Constraints \eqref{P1_C11} and \eqref{P1_C12} imply $0 \leq \alpha_{\rm n}^{\mathrm{HD}} \leq \frac{1}{2}$ In addition, constraints \eqref{P1_C12} and \eqref{P1_C15} imply $\alpha_{\rm n}^{\mathrm{HD}} \leq \frac{\gamma_{\rm bn} - t_{\rm f}^{\rm HD}}{\gamma_{\rm bn} \left(t_{\rm f}^{\rm HD}+1\right)}$. Therefore, it can be concluded that $\alpha_{\rm n}^{\mathrm{HD}}$ should satisfy $\alpha_{\min}^{\rm HD} \leq \alpha_{\rm n}^{\mathrm{HD}} \leq  \alpha_{\max}^{\rm HD}$, where $\alpha_{\min}^{\rm HD}$ and $\alpha_{\max}^{\rm HD}$ are expressed as shown in \eqref{HD_bounds} in \textbf{\textit{Theorem 1}}. Hence, for problem $\mathrm{PC-HD}$ to be feasible, the condition $\alpha_{\min}^{\rm HD} \leq \alpha_{\max}^{\rm HD}$ should be satisfied. On the other hand, one can remark that the lowest value of $\beta^{\rm HD}$ that satisfies  constraint \eqref{P1_C16} is reached at the lowest feasible value of $\alpha_{\rm n}^{\mathrm{HD}}$ and the highest feasible value of $\alpha_{\rm f}^{\mathrm{HD}}$. In this context, based on constraint \eqref{P1_C14}, the lowest feasible value of $\alpha_{\rm n}^{\mathrm{HD}}$ is $\alpha_{\rm n}^{\mathrm{HD}} = \alpha_{\min}^{\rm HD} = \frac{t^{\rm HD}}{\gamma_{\rm bn}}$ and, based on constraint \eqref{P1_C12}, the highest feasible value of $\alpha_{\rm f}^{\mathrm{HD}}$ is $\alpha_{\rm n}^{\mathrm{HD}} = 1 - \alpha_{\min}^{\rm HD} = 1-\frac{t^{\rm HD}}{\gamma_{\rm bn}}$. Therefore, the lowest feasible value of $\beta^{\rm HD}$ is given by $ \beta_{\min}^{\rm HD}$ in \eqref{HD_bounds}. Consequently, for problem $\mathrm{PC-HD}$ to be feasible, the condition $\beta_{\min}^{\rm HD} \leq \beta_{\max}^{\rm HD}$ should be also satisfied, where $\beta_{\max}^{\rm HD} = 1$, which completes the proof.
\section{Proof of Theorem 2}
In this section, we present the proof of \textbf{\textit{Theorem 2}}. With the goal of minimizing the total transmit power of the C-NOMA cellular system, two directions can be considered. The first direction consists of minimizing the BS transmit power first and then minimizing the one of the near UE, while the second direction consists of minimizing the transmit power of the near UE first and then minimizing the one of the BS. Let us start with the first direction. As discussed in the previous appendix, the lowest feasible value of $\alpha_{\rm n}^{\rm HD}$ is $\alpha_{\rm n,1}^{\rm HD} = \alpha_{\min}^{\rm HD}$. Based on this, we can see from constraints \eqref{P1_C11} and \eqref{P1_C15} that the lowest feasible value of $\alpha_{\rm f}^{\rm HD}$ is $\alpha_{\rm f,1}^{\rm HD} = \max\left(\alpha_{\min}^{\rm HD},\alpha_{\min}^{\rm HD} t_{\rm f}^{\rm HD} + \frac{t_{\rm f}^{\rm HD}}{\gamma_{\rm n}}\right)$. Therefore, by injecting these two values in constraint \eqref{P1_C16}, we conclude that the lowest value of $\beta^{\rm HD}$ is $\beta_1^{\rm HD} = \frac{1}{\gamma_{\rm d}} \left(t_{\rm f}^{\rm HD} - \frac{\alpha_{\rm f,1}^{\rm HD} \gamma_{\rm f}}{\alpha_{\rm n,1}^{\rm HD}\gamma_{\rm f}+1} \right)$. Now, considering the second direction, the lowest feasible value of $\beta^{\rm HD}$ is $\beta_1^{\rm HD} = 0$. on the other hand, recall that the lowest feasible value of $\alpha_{\rm n}^{\rm HD}$ is $\alpha_{\rm n,2}^{\rm HD} = \alpha_{\min}^{\rm HD}$. Therefore, by injecting these two values in constraint \eqref{P1_C16}, we find that the lowest feasible value of $\alpha_{\rm f}^{\rm HD}$ is $\alpha_{\rm f,2}^{\rm HD} = \frac{\left(\alpha_{\min}^{\rm HD}\gamma_{\rm bn} + 1\right)t_{\rm f}^{\rm HD}}{\gamma_{\rm bn}}$. Finally, the optimal power control scheme is the one among the above two schemes that minimize the total transmit power of the HD C-NOMA system, which completes the proof.
\section{Proof of Theorem 3}
In this section, we present the proof of \textbf{\textit{Theorem 3}}. At the beginning, constraints \eqref{P3_C1}-\eqref{P3_C6} can be expressed as
\begin{subequations}
\begin{align}
&0 \leq \alpha_{\rm n}^{\mathrm{FD}} \leq \alpha_{\rm f}^{\mathrm{FD}} \leq 1,\label{P3_C11}\\
&\alpha_{\rm f}^{\mathrm{FD}}  \leq 1-\alpha_{\rm n}^{\mathrm{FD}},\label{P3_C12} \\
&0 \leq \beta^{\mathrm{FD}} \leq 1,\label{P3_C13} \\
&\alpha_{\rm n}^{\rm FD} \geq \frac{\gamma_{\rm SI}}{ \gamma_{\rm bn}} \beta^{\mathrm{FD}} +  \frac{t_{\rm n}^{\rm FD}}{\gamma_{\rm bn}},\label{P3_C14} \\ 
&\alpha_{\rm f}^{\rm FD} \geq \alpha_{\rm n}^{\rm FD} t_{\rm f}^{\rm FD} + \frac{\gamma_{\rm SI} t_{\rm f}^{\rm FD}}{\gamma_{\rm bn}} \beta^{\mathrm{FD}} + \frac{t_{\rm f}^{\rm FD}}{\gamma_{\rm bn}},\label{P3_C15} \\ 
&\alpha_{\rm f}^{\rm FD} \geq \alpha_{\rm n}^{\rm FD} t_{\rm f}^{\rm FD} - \frac{ \gamma_{\rm d} t_{\rm f}^{\rm FD}}{ \gamma_{\rm bf}} \beta^{\mathrm{FD}} + \frac{t_{\rm f}^{\rm FD}}{\gamma_{\rm bf}} \label{P3_C16}.
\end{align}
\end{subequations}
Constraints \eqref{P3_C11} and \eqref{P3_C12} imply $0 \leq \alpha_{\rm n}^{\mathrm{FD}} \leq \frac{1}{2}$. Let $B_1$ be the bound whose expression is given in constraint \eqref{P3_C14}. In addition, constraints \eqref{P3_C12} and \eqref{P3_C15} imply 
\begin{equation}
    \label{P3_C18}
    B_2:\,\,\alpha_{\rm n}^{\mathrm{HD}} \leq \frac{\gamma_{\rm bn} - t_{\rm f}^{\rm HD} - P_{\rm n} \gamma_{\rm SI} t_{\rm f}^{\rm HD} \beta^{\mathrm{FD}} }{\gamma_{\rm bn} \left(t_{\rm f}^{\rm HD}+1\right)}.
\end{equation}
Moreover, based on constraints \eqref{P3_C12} and \eqref{P3_C16}, we obtrain
\begin{equation}
    \label{P3_C19}
    B_3:\,\,\alpha_{\rm n}^{\mathrm{HD}} \leq \frac{\gamma_{\rm bn} - t_{\rm f}^{\rm HD} + P_{\rm n} \gamma_{\rm d} \beta^{\mathrm{FD}} }{\gamma_{\rm bf} \left(t_{\rm f}^{\rm HD}+1\right)}.
\end{equation}
The bounds $B_1$, $B_2$ and $B_3$ are, each, a function of the variables
$\left(\alpha_{\rm n}^{\rm FD},\beta^{\rm FD}\right)$. Therefore, problem $\mathrm{PC-FD}$ is feasible if and only if the bound $B_1$ can be lower than the bounds $B_2$ and $B_3$ simultaneously within the region $\left(\alpha_{\rm n}^{\rm FD},\beta^{\rm FD}\right) \in \left[0,\frac{1}{2} \right] \times \left[0,1 \right]$. The condition $B_1 \leq B_2$ implies that $\beta^{\rm FD}$ should satisfy $\beta^{\rm FD} \leq \frac{\gamma_{\rm n} - t_{\rm f}^{\rm FD} - t_{\rm n}^{\rm FD}\left(1 + t_{\rm f}^{\rm FD} \right)}{\gamma_{\rm SI}t_{\rm n}^{\rm FD}\left(1 + t_{\rm f}^{\rm FD} \right) + t_{\rm f}^{\rm FD}}$. On the other hand, the condition $B_1 \leq B_3$ implies that $\beta^{\rm FD}$ should satisfy $\beta^{\rm FD} \geq \frac{ \gamma_{\rm bf} - t_{\rm f}^{\rm FD} - \frac{\gamma_{\rm bf}\left(1 + t_{\rm f}^{\rm FD} \right)t_{\rm n}^{\rm FD}}{\gamma_{\rm n}}}{c_1-c_2}$ if $c_1 \leq c_2$ and $\beta^{\rm FD} \leq \frac{\gamma_{\rm bf} - t_{\rm f}^{\rm FD} - \frac{\gamma_{\rm bf}\left(1 + t_{\rm f}^{\rm FD} \right)t_{\rm n}^{\rm FD}}{\gamma_{\rm n}}}{c_1-c_2}$ if $c_1 \geq c_2$. Therefore, along with the fact that $\beta^{\rm FD} \in [0,1]$, we obtain the first condition in \textbf{\textit{Theorem 3}}. Furthermore, since $B_1$ should be lower than the bounds $B_2$ and $B_3$ simultaneously within the region $\left(\alpha_{\rm n}^{\rm FD},\beta^{\rm FD}\right) \in \left[0,\frac{1}{2} \right] \times \left[0,1 \right]$, we conclude that $\frac{\gamma_{\rm SI} t_{\rm n}^{\rm FD}}{ \gamma_{\rm bn}} \beta_{\min}^{\rm FD} + \frac{t_{\rm n}^{\rm FD}}{\gamma_{\rm bn}}$ should be lower than $\frac{1}{2}$, which is the second condition of \textbf{\textit{Theorem 3}} and this completes the proof.
\section{Proof of Theorem 4}
In this section, we present the proof of \textbf{\textit{Theorem 4}}. Since the objective function of problem $\mathrm{PC-FD}$, the optimal solution is one of the intersection points of the boundaries of its feasibility region. The point with coordinates $\left(\alpha_{\rm c}^{\rm FD},\beta_{\rm c}^{\rm FD}\right)$, defined in \eqref{eq:vars} represents the intersection point of the bounds $B_2$ and $B_3$. Therefore, depending on whether $\beta_{\rm c}^{\rm FD}$ is within the feasibility interval $\left[\beta_{\min}^{\rm FD}, \beta_{\max}^{\rm FD} \right]$, this intersection point is transformed into the point $\left(\alpha_{0}^{\rm FD}, \beta_{0}^{\rm FD}\right)$ in \eqref{eq:transformation} in order to verify if this point should be considered as a candidate solutions or no. Based on this, the intersection points of the the boundaries of the feasibility region of problem $\mathrm{PC-FD}$ are given in \eqref{eq:intersection_points}. Afterwards, for each intersection point $\left(\alpha_{\rm n}^{\rm FD},\beta^{\rm FD}\right)$, the corresponding $\alpha_{\rm f}^{\rm FD}$ can be obtained as shown in \eqref{eq:power_far}, which in turn constructs a candidate solution $\left(\alpha_{\rm n}^{\rm FD},\alpha_{\rm f}^{\rm FD},\beta^{\rm FD} \right)$. Consequently, the optimal solution of problem $\mathrm{PC-FD}$ is the candidate solution that produces the lowest objective function as presented in \textbf{\textit{Theorem 4}}, which completes the proof.
\bibliographystyle{IEEEtran}
\bibliography{IEEEabrv,reference}
\end{document}